\documentclass{article}

\usepackage{amssymb} 
\usepackage{amsfonts}
\usepackage{amsmath}
\usepackage{comment}

\usepackage{savesym}
\savesymbol{vartriangleleft}
\savesymbol{vartriangleright}

\usepackage{xspace}
\usepackage{enumerate}

\usepackage{hyperref}
\usepackage{cleveref}
\usepackage{amsthm}
\usepackage{amssymb}
\usepackage{cmll}
\usepackage[pdftex,noxy]{virginialake}
\usepackage{url}
\usepackage{tikz}
\usetikzlibrary{arrows}
\usetikzlibrary{snakes}
\usetikzlibrary{shapes.geometric}

\usepackage{rotate}
\usepackage{txfonts}

\usepackage{arydshln}

\usepackage{thmtools,thm-restate}

\usepackage{adjustbox}

\def\ltens{\mathbin\varotimes}

\def\limp{\multimap}

\def\lcoseq{\mathbin\origvartriangleright}
\def\imp{\supset}
\def\limp{\multimap}
\newcommand{\lbox}{\Box}
\newcommand{\ldia}{\Diamond}

\newcommand{\provesym}{\tikz[baseline=-.65ex]{\draw[very thick] (0,0)--(.4,0);\draw[thin] (0,-.1)--(0,.1);}}
\newcommand{\provevia}[1]{\mathrel{\stackrel{\scriptscriptstyle#1}{\provesym}}}

\def\deep#1{#1_\downarrow}

\def\N{\mathbb N}

\def\atoms{\mathcal A}

\newcommand\gclr{\color{blue}} 

\newcommand\iclr{\color{\implcolor}} 
\newcommand\mclr{\color{\modcolor}} 
\newcommand\sclr{\color{\skewcolor}}

\def\gG{\mathcal G}
\def\gH{\mathcal H}
\def\gF{\mathcal F}
\def\gJ{\mathcal J}
\def\gK{\mathcal K}
\def\gI{\mathcal I}

\def\dD{\mathfrak D}

\newcommand{\cgF}{{\gclr\gF}}
\newcommand{\cgG}{{\gclr\gG}}
\newcommand{\cgH}{{\gclr\gH}}
\newcommand{\cgK}{{\gclr\gK}}
\newcommand{\cgJ}{{\gclr\gJ}}
\newcommand{\cgI}{{\gclr\gI}}

\def\set#1{\{#1\}}
\def\Set#1{\begin{Bmatrix}{#1}\end{Bmatrix}}
\def\sizeof#1{|#1|}
\def\tuple#1{\langle#1\rangle}

\newcommand{\intset}[2]{\set{#1, \dots, #2}}

\def\whiter#1{#1^\circ}
\def\blackr#1{#1^\bullet}

\def\rlimpr{\whiter\limp}
\def\llimpr{\blackr\limp}
\def\rltensr{\whiter \ltens}
\def\lltensr{\blackr \ltens}

\def\AXrule{\mathsf {AX}}
\def\Wrule{\mathsf {W}}
\def\Crule{\mathsf {C}}
\def\axrule{\mathsf {ax}}
\def\cutr{\mathsf {cut}}

\def\limprule{{\imp}^\mathsf L}
\def\rimprule{{\imp}^\mathsf R}
\def\llandrule{{\land}^\mathsf L}
\def\rlandrule{{\land}^\mathsf R}

\def\krule{\mathsf {k}}

\def\kbrule{\krule^\lbox}
\def\kdrule{\krule^\ldia}
\def\kurule{\krule^\unit}
\def\drule{\mathsf {d}}

\def\Krule{\mathsf {K}}
\def\Drule{\mathsf {D}}

\def\Kbrule{\Krule^\lbox}
\def\Kdrule{\Krule^\ldia}
\def\Kurule{\Krule^\unit}

\def\wrule{{\mathsf {w}}}
\def\worule{\blackr\wrule}
\def\wtrule{\wrule^{\ltens}}
\def\wirule{\wrule^{\limp}}
\def\wdrule{\mathsf {w}^{\ldia}}
\def\crule{\blackr{\mathsf {c}}}

\newcommand{\diredge}[1][]{\mkern1mu\mathord{\stackrel{#1}{\rightarrow}}\mkern1mu}
\newcommand{\codiredge}[1][]{\mkern1mu\mathord{\stackrel{#1}{\leftarrow}}\mkern1mu}
\newcommand{\sdiredge}[1][]{\mkern1mu\mathord{\stackrel{#1}{\rightarrow\!\!\!\!\!\!\leftarrow}}\mkern1mu}

\newcommand{\iedge}[1][]{\mkern1mu\mathord{\stackrel{#1}{{\color{\implcolor}\rightarrow}}}\mkern1mu}
\newcommand{\oedge}[1][]{\mkern1mu\mathord{\stackrel{#1}{{\color{\implcolor}\rightarrow}_\bullet}}\mkern1mu}
\newcommand{\iiedge}[1][]{\mkern1mu\mathord{\stackrel{#1}{{\color{\implcolor}\leftarrow}}}\mkern1mu}

\newcommand{\medge}[1][]{\mkern1mu\mathord{\stackrel{#1}{{\color{\modcolor}\rightsquigarrow}}}\mkern1mu}
\newcommand{\imedge}[1][]{\mkern1mu\mathord{\stackrel{#1}{{\color{\modcolor}\leftsquigarrow}}}\mkern1mu}
\newcommand{\axlink}[1][]{\mkern1mu\mathord{\stackrel{#1}{{\gclr\rightharpoonup}}}\mkern1mu}

\newcommand{\axeq}[1][]{\mkern1mu\mathord{\stackrel{#1}{{\gclr\sim}}}\mkern1mu}
\newcommand{\naxeq}[1][]{\mkern1mu\mathord{\stackrel{#1}{{\gclr\not\sim}}}\mkern1mu}

\newcommand{\omedge}[1][]{\mkern1mu\mathord{\stackrel{#1}{{\color{\implcolor}\rightsquigarrow}}}\mkern1mu_{\partial}}

\newcommand{\nmedge}[1][]{\!\not\!\!\!\!\!\medge[#1]}

\newcommand{\niedge}[1][]{\not\!\!\iedge[#1]}
\newcommand{\diedge}[1][]{\mkern1mu\mathord{\stackrel{#1}{{\color{\implcolor}\rightarrow\!\!\!\!\!\!\leftarrow}}}\mkern1mu}
\newcommand{\dmedge}[1][]{\mkern1mu\mathord{\stackrel{#1}{{\color{\modcolor}\leftrightsquigarrow}}}\mkern1mu}
\newcommand{\daxlink}[1][]{\mkern1mu\mathord{\stackrel{#1}{{\gclr\leftharpoondown\hspace{-8pt}\rightharpoonup}}}\mkern1mu}

\newbox\auxbox
\setbox\auxbox=\hbox{$\mathord\curvearrowright$}
\def\auxsym{\rlap{\copy\auxbox}\raise.3ex\hbox{\copy\auxbox}}

\def\modsym{\ssbox\ssdia}

\newcommand{\vertices}[1][]{V_{#1}}
\newcommand{\avertices}[1][]{V^{\atoms}_{#1}}
\newcommand{\bvertices}[1][]{V^\ssbox_{#1}}
\newcommand{\dvertices}[1][]{V^{\ssdia}_{#1}}
\newcommand{\mvertices}[1][]{V^{\modsym}_{#1}}

\tikzstyle{edgestyle}=[>=stealth,overlay,remember picture,thin, opacity=1]

\newcommand{\flowvertex}[2]{\mathord{%
    \tikz[remember picture,baseline=(#2\nodecode.base)]%
    \node[inner sep=0pt](#2\nodecode){$#1\strut$};}}

\newcommand{\va}[1]{\flowvertex{a}{a#1}}

\newcommand{\vb}[1]{\flowvertex{b}{b#1}}

\newcommand{\vv}[1]{\flowvertex{v}{v#1}}
\newcommand{\vw}[1]{\flowvertex{w}{w#1}}
\newcommand{\vu}[1]{\flowvertex{u}{u#1}}

\newcommand{\vc}[1]{\flowvertex{c}{c#1}}

\newcommand{\vd}[1]{\flowvertex{d}{d#1}}

\newcommand{\ve}[1]{\flowvertex{e}{e#1}}
\newcommand{\vf}[1]{\flowvertex{f}{f#1}}

\newcommand{\vlbox}[1]{\flowvertex{\lbox}{box#1}}
\newcommand{\vldia}[1]{\flowvertex{\ldia}{dia#1}}
\newcommand{\vunit}[1]{\flowvertex{\unit}{unit#1}}

\def\rootsym{{\iclr\blacktriangleright}}

\newcommand{\vrroot}[1]{\flowvertex{\iclr\blacktriangleright}{m#1}}
\newcommand{\vlroot}[1]{\flowvertex{\iclr\blacktriangleright}{m#1}}

\newcommand{\vimp}[1]{\flowvertex{\imp}{imp#1}}
\newcommand{\vand}[1]{\flowvertex{\land}{and#1}}

\definecolor{burgundy}{rgb}{0.5, 0.0, 0.13}
\definecolor{britishracinggreen}{rgb}{0.0, 0.26, 0.15}

\def\implcolor{burgundy}
\def\modcolor{britishracinggreen}			
\def\indcolor{violet}
\def\linkcolor{blue}

\def\skewcolor{magenta}

\tikzstyle{treestyle}=[>=stealth,overlay,remember picture,thin, draw,fill,opacity=1]  
\tikzstyle{arenastyle}=[>=stealth,overlay,remember picture,thin, draw=\implcolor,fill=\implcolor,opacity=1]  
\tikzstyle{coarenastyle}=[densely dotted,>=stealth,overlay,remember picture, draw=\indcolor,fill=\indcolor,opacity=1]  
\tikzstyle{modstyle}=[>=stealth,overlay,remember picture,thin, draw=\modcolor,fill=\modcolor,opacity=1]
\tikzstyle{omodstyle}=[>=stealth,overlay,remember picture,thin, draw=\implcolor,fill=\implcolor,opacity=1]
\tikzstyle{linkstyle}=[>=left to,overlay,remember picture,thick, dashed, draw=\linkcolor,opacity=1]  
\tikzstyle{skewstyle}=[densely dotted,>=stealth,overlay,remember picture,thin, draw=\skewcolor,opacity=1]

\newcommand{\skewedge}[2]{%
	\tikz[skewstyle]%
	\draw[->] (#1\nodecode) --  (#2\nodecode);
}

\newcommand{\skewedges}[1]{
	\foreach \aaa/\bbb in {#1} {\skewedge{\aaa}{\bbb}}
}

\newcommand{\treeedge}[2]{%
	\tikz[treestyle]%
	\draw[-] (#1\nodecode) -- (#2\nodecode);
}
\newcommand{\treeedges}[1]{
	\foreach \aaa/\bbb in {#1} {\treeedge{\aaa}{\bbb}}
}

\newcommand{\dedge}[2]{%
  \tikz[arenastyle]%
  \draw[->] (#1\nodecode) -- (#2\nodecode);
}
\newcommand{\dedges}[1]{
  \foreach \aaa/\bbb in {#1} {\dedge{\aaa}{\bbb}}
}

\newcommand{\multidedges}[2]{
\foreach \aaa in {#1}\foreach\bbb in {#2}{\dedge{\aaa}{\bbb}}
}

\newcommand{\modedge}[2]{%
  \tikz[modstyle]%
  \draw[->,decorate,decoration=snake,segment amplitude=1pt,segment length=5pt] (#1\nodecode) -- (#2\nodecode);
}
\newcommand{\modedges}[1]{
  \foreach \aaa/\bbb in {#1} {\modedge{\aaa}{\bbb}}
}

\newcommand{\multimodedges}[2]{
	\foreach \aaa in {#1}\foreach\bbb in {#2}{\modedge{\aaa}{\bbb}}
}

\newcommand{\axedge}[2]{%
  \tikz[linkstyle]%
  \draw[-] (#1\nodecode) -- (#2\nodecode);
}

\newcommand{\axedges}[1]{
  \foreach \aaa/\bbb in {#1} {\axedge{\aaa}{\bbb}}
}

\newcommand{\axlbent}[2]{%
  \tikz[linkstyle]%
  \draw[-] (#1\nodecode) to [bend left] (#2\nodecode);
}

\newcommand{\axlbents}[1]{
  \foreach \aaa/\bbb in {#1} {\axlbent{\aaa}{\bbb}}
}

\newcommand{\axllbent}[2]{%
	\tikz[linkstyle]%
	\draw[-] (#1\nodecode) to [bend left=10] (#2\nodecode);
}

\newcommand{\axllbents}[1]{
	\foreach \aaa/\bbb in {#1} {\axllbent{\aaa}{\bbb}}
}

\newcommand{\dlbent}[2]{%
	\tikz[arenastyle]%
	\draw[->] (#1\nodecode) to [bend left] (#2\nodecode);
}

\newcommand{\drbent}[2]{%
	\tikz[arenastyle]%
	\draw[->] (#1\nodecode) to [bend right] (#2\nodecode);
}

\newcommand{\axrbent}[2]{%
  \tikz[linkstyle]%
  \draw[-] (#1\nodecode) to [bend right] (#2\nodecode);
}

\newcommand{\axrbents}[1]{
  \foreach \aaa/\bbb in {#1} {\axrbent{\aaa}{\bbb}}
}

\def\nodecode{}

\theoremstyle{plain}
\newtheorem{theorem}{Theorem}

\newtheorem{corollary}  [theorem]{Corollary}
\newtheorem{lemma}      [theorem]{Lemma}
\newtheorem{proposition}[theorem]{Proposition}

\theoremstyle{definition}
\newtheorem{definition}[theorem]{Definition}
\newtheorem{example}[theorem]{Example}

\newtheorem{notation}[theorem]{Notation}
\newtheorem{remark}[theorem]{Remark}
\def\IMLL{\mathsf{IMLL}}
\def\IMLLX{\IMLL\mbox{-}\X}
\def\IMLLK{\IMLL\mbox{-}\CK}
\def\IMLLD{\IMLL\mbox{-}\CD}

\def\X{\mathsf{X}}

\def\K{\mathsf{K}}
\def\D{\mathsf{D}}

\def\CK{\mathsf{CK}}
\def\CD{\mathsf{CD}}

\def\CX{\mathsf{CX}}

\def\kax{\mathsf{k}}

\def\Ip{\mathsf{LI}}
\def\IpX{\mathsf{L}\X}
\def\IpK{\mathsf{L}\CK}
\def\IpD{\mathsf{L}\CD}

\newcommand*\halfcirc[1][1ex]{%
	\begin{tikzpicture}
		\draw[fill] (0,0)-- (90:#1) arc (90:270:#1) -- cycle ;
		\draw (0,0) circle (#1);
	\end{tikzpicture}
}

\def\funnysymbol{{\resizebox{3.5pt}{!}{\halfcirc}}}
\def\POL#1{{#1}^\funnysymbol}

\def\DOWN#1{#1_\downarrow}

\def\PIp{\POL \Ip}
\def\PX{\POL \IpX}
\def\PK{\POL\IpK}
\def\PD{\POL\IpD}

\def\PIMLL{\POL\IMLL}
\def\PIMLLX{\POL\IMLLX}
\def\PIMLLK{\POL\IMLLK}
\def\PIMLLD{\POL\IMLLD}

\def\mathacronym{\mathsf}

\def\dagmacro{$\mathacronym{dag}$}
\def\mamacro{$\mathacronym{MA}$}
\def\ddagmacro{$\operatorname{2-\mathacronym{dag}}$}

\def\ma{\mamacro\xspace}
\def\mas{\mamacro s\xspace}

\def\dag{\dagmacro\xspace}
\def\dags{\dagmacro s\xspace}
\def\ddag{\ddagmacro\xspace}
\def\ddags{\ddagmacro s\xspace}

\def\gsum{+}
\def\gimp{{\iclr-\hspace{-2pt}\lcoseq}}
\def\gmod{{\mclr\sim\hspace{-4.2pt}\lcoseq}}
\def\unit{\bot}

\def\labsym{\ell}
\def\lab#1{\labsym(#1)}

\newcommand{\arof}[1]{{\llbracket#1\rrbracket}}
\newcommand{\ssdia}{{\diamond}}
\newcommand{\ssbox}{{\scriptscriptstyle\Box}}
\newcommand{\singlevertex}[1][]{{#1}}

\makeatletter
\newcommand{\oset}[3][.5ex]{%
	\mathrel{\mathop{#3}\limits^{
			\vbox to#1{\kern-2\ex@
				\hbox{$\scriptstyle#2$}\vss}}}}
\makeatother

\def\spalletto#1{\lfloor#1\rfloor}
\newcommand{\dirrof}[1][]{\oset{\diredge}{R}_{\!#1}}
\newcommand{\irof}[1][]{\oset{\iedge}{R}_{\!#1}}
\newcommand{\mrof}[1][]{\oset{\medge}{R}_{\!#1}}

\def\isdirroot#1{#1\not\!\!\!\diredge \;}

\def\isdirroot#1{#1\not\!\!\!\diredge}

\newcommand{\icone}[2]{\set{#1\iedge ^{#2}}}
\def\mconeof#1{\oset{\medge}{\mathsf C}\!\!({#1})}
\def\pmv#1{{\mclr\hat{\color{black}#1}}}

\def\pformula{$\mathsf P$-formula\xspace}
\def\pformulas{$\mathsf P$-formulas\xspace}

\def\ftree#1{\mathcal T_#1}

\def\cons#1{\{#1\}}

\def\labelset{\mathcal L}
\def\normto{\rightsquigarrow}

\def\oddsym{\bullet}
\def\evensym{\circ}

\def\depsym{{\iclr d}}
\def\dep#1{\depsym(#1)}
\def\iseven#1{#1^{\evensym}}
\def\isodd#1{#1^\oddsym}
\def\linkgraph#1{\overset{\color{blue} \curvearrowright}{#1}}

\def\sameframesymbol{\;\mathord{{\mclr{\downarrow}}\!\!\!\!\raisebox{2pt}{$\axeq$}}}

\newcommand{\sameframe}[2]{#1 \sameframesymbol #2}

\def\der#1{\partial({#1})}

\def\id{1}

\def\cf{{\sclr f}}
\def\cg{{\sclr g}}

\def\cid{{\sclr \id}}
\def\cempty{{\sclr \emptyset}}
\newcommand{\cpair}[2]{{{\sclr [{\color{black}#1}, {\color{black}#2}]}}}

\newcommand{\conj}[1][]{\curlywedge_{#1}}
\newcommand{\nconj}[1][]{\not\!\!\curlywedge_{#1}}
\newcommand{\disj}[1][]{\curlyvee_{#1}}
\newcommand{\ndisj}[1][]{\not\mkern-4mu\curlyvee_{#1}}

\def\feq{\oset{\mathsf f } \sim}

\def\strat{\mathcal S}

\def\evensym{\circ}
\def\oddsym{\bullet}
\def\evenmove{$\evensym$-move\xspace}
\def\oddmove{$\oddsym$-move\xspace}
\def\evenmoves{$\evensym$-moves\xspace}
\def\oddmoves{$\oddsym$-moves\xspace}

\def\jpath{\mathsf p}
\def\jpathi#1{\jpath_{#1}}	

\newcommand{\stratlink}[2]{\axlink[{#1}_{#2}]}
\newcommand{\dstratlink}[2]{\daxlink[{#1}_{#2}]}
\newcommand{\strateq}[2]{\axeq[{#1}_{#2}]}
\newcommand{\strateqone}[2]{\axeq[{#1}_{#2}]_1}

\newcommand{\nstrateq}[2]{\naxeq[{#1}_{#2}]}	
\def\view{\mathsf{p}}
\def\viewi#1{\view_{#1}}
\def\aview{\tilde{\view}}
\def\aviewi#1{\tilde{\view}_{#1}}
\def\astrat{\tilde \strat}

\def\fviewof#1{\mathcal F(#1)}

\def\unchecked#1{{#1}^\downharpoonright}

\def\ICP{$\mathsf{ICP}$\xspace}
\def\ICPs{$\mathsf{ICP}$s\xspace}
\def\CPs{$\mathsf{CP}$s\xspace}
\def\XICP{$\X$-\ICP}
\def\XICPs{$\X$-\ICPs}
\def\CKICP{$\CK$-\ICP}
\def\CKICPs{$\CK$-\ICPs}

\def\CDICPs{$\CD$-\ICPs}

\def\WIS{$\mathsf{WIS}$\xspace}
\def\WISs{$\mathsf{WIS}$s\xspace}
\def\XWIS{$\X$-\WIS}
\def\XWISs{$\X$-\WISs}
\def\CKWIS{$\CK$-\WIS}
\def\CKWISs{$\CK$-\WISs}

\def\CDWISs{$\CD$-\WISs}

\begin{document}

\title{Towards a Denotational Semantics for Proofs in Constructive Modal Logic}

\author{Matteo Acclavio \and Davide Catta \and Lutz Stra\ss burger}
\date{}

\maketitle

\begin{abstract}	
In this paper we provide two new semantics for proofs in the constructive modal logics $\CK$ and $\CD$.

The first semantics is given by extending the syntax of combinatorial proofs for propositional intuitionistic logic, in which proofs are factorised in a linear fragment (arena net) and a parallel weakening-contraction fragment (skew fibration).
In particular we provide an encoding of modal formulas by means of directed graphs (modal arenas),
and 
an encoding of linear proofs as modal arenas equipped with vertex partitions satisfying topological criteria.

The second semantics is given by means of winning innocent strategies of a two-player game over modal arenas.
This is given by extending the Heijltjes-Hughes-Stra\ss burger correspondence between  intuitionistic combinatorial proofs and winning innocent strategies in a Hyland-Ong arena.
Using our first result, we provide a characterisation of winning strategies for games on a modal arena corresponding to proofs with modalities.

\end{abstract}

\section{Introduction}

	Semantics is the area of logic concerned with specifying the meaning of the logical constructs. 
	We distinguish between two main kind of semantic approach to logic. 
	The first, the model-theoretic approach, is concerned with specifying the meaning of formulas in terms of truth in some model. 
	The second, the denotational semantic approach, is concerned with specifying the meaning of proofs of the logic under a compositional point of view.
	Proofs are interpreted as mathematical objects called denotation, and the meaning of composed proofs is obtained by composing denotations.

Modal logics are extensions of classical logic making use of \emph{modalities} to qualify the truth of a judgement.
According with the interpretation of such modalities, modal logics find applications, for example, in knowledge representation~\cite{knowledgerepresentation}, artificial intelligence~\cite{MEYERAI} and verification~\cite{horne:quasi}.
More precisely, modal logics are obtained by extending classical logic with a modality operator $\lbox$ (together with its dual operator $\ldia$), which are usually interpreted as \emph{necessity} (respectively  \emph{possibility}).

When we move from the classical to the intuitionistic setting
we are forced to make some choices since there are many different flavours of ``intuitionistic modal logics''~(see, e.g., \cite{fitch1948intuitionistic,prawitz1965proof,plotkin1986framework,simpson:phd,bie:paiva:intuitionistic,davies2001modal}).
This range of possible extensions of the intuitionistic logic depends on the fact that the classical $\kax$-axiom $\lbox(A \imp B) \imp (\lbox A \imp \lbox B)$  is no longer sufficient to express the behaviour of the modality~$\ldia$ as it is no longer the dual of $\lbox$.
We here consider the minimal approach and only add the axiom $\lbox(A \imp B) \imp (\ldia A \imp \ldia B)$, leading to what in the literature is now called \emph{constructive modal logics} \cite{prawitz1965proof,bie:paiva:intuitionistic,heilala2007bidirectional,mend:multimodal,fairtlough1997propositional,kojima2012semantical}. 

	Both the denotational approach and the model-theoretic approach have been developed in the literature on constructive modal logics. 
	One of the desired feature of denotational models is full completeness: in a full complete model every denotation is the interpretation of some proof. 
	Reasoning about the property of full complete models allows one to have a syntax-free characterization of the property of proofs. We say that  a denotational model is \emph{concrete} if its elements are not  obtained by the quotient on proofs induced by cut-elimination.	 
To our knowledge, the only  full complete denotational model for this logic is 
not concrete since 
defined by the quotient of their  $\lambda$-calculi with respect to $\beta$-reduction~\cite{bel:deP:rit:extended,bie:paiva:intuitionistic}.

The purpose of this paper is to lay the foundations 
for a concrete denotational full complete model in terms of 
a \emph{game semantics}~\cite{abr:full,hyl:ong:PCF,mcc:FPC} for this logic
by providing a definition of proofs denotations.
Game semantics is a denotational semantics where proofs are denoted by  winning strategies for a two-player game.
In \cite{ICP} it is shown how the syntax of intuitionistic combinatorial proofs (or \ICPs), a graphical proof system for propositional intuitionistic logic, provides some intuitive insights about the innocent winning strategies (or \WISs) in a Hyland-Ong arena~\cite{hyl:ong:PCF,mur:ong:evolving}.
In order to define \WIS for constructive modal logics 
we extend this correspondence. 
For this, we first provide the definition \ICPs for these logics  as shown in \Cref{fig:intro}.

\begin{figure*}[t!]
	\begin{adjustbox}{max width=\textwidth}
		$\begin{array}{c|c}
			\vlderivation{
			\vliq{}{\rimprule}{ \lbox (b\imp b)\imp a) \imp( \ldia c\imp \ldia(a\land a))}{
				\vlin{}{\Kdrule}{\ldia c, \lbox ((b\imp b)\imp a)\vdash \ldia(a\land a)}{
					\vlin{}{\Wrule}{c, (b\imp b)\imp a\vdash a\land a}{
						\vlin{}{\Crule}{(b\imp b)\imp a\vdash a\land a}{
							\vliin{}{\rlandrule}{(b\imp b)\imp a , (b\imp b)\imp a \vdash a\land a}
							{\vliin{}{\limprule}{(b\imp b)\imp a \vdash a}{\vlin{}{\rimprule}{\vdash b \imp b}{\vlin{}{\AXrule}{b \vdash b}{\vlhy{}}}}{\vlin{}{\AXrule}{a \vdash a}{\vlhy{}}}}
							{\vliin{}{\limprule}{(b\imp b)\imp a \vdash a}{\vlin{}{\rimprule}{\vdash b \imp b}{\vlin{}{\AXrule}{b \vdash b}{\vlhy{}}}}{\vlin{}{\AXrule}{a \vdash a}{\vlhy{}}}}
						}
					}
			}}
		}
		&
		\begin{array}{cc}
			\vlderivation{
				\vliq{}{\rimprule}{ \lbox (b\imp b)\imp a) \imp( \ldia c\imp \ldia(a\land a))}{
				\vlin{}{\Kdrule}{\ldia c, \lbox ((b\imp b)\imp a)\vdash \ldia(a\land a)}
				{\vliin{}{\limprule}{c , (b \imp b)\imp a \vdash a \land a}
					{\vlin{}{\rimprule}{c \vdash b \imp b}{\vlin{}{\Wrule}{c,b \vdash b}{\vlin{}{\AXrule}{b \vdash b}{\vlhy{}}}}}
					{\vlin{}{\Crule}{a \vdash a\land a}{
							\vliin{}{\llandrule}{a, a \vdash a\land a}
							{\vlin{}{\AXrule}{a \vdash a}{\vlhy{}}}
							{\vlin{}{\AXrule}{a \vdash a}{\vlhy{}}}
						}
					}}}
			}
			&
			\vlderivation{
				\vliq{}{\rimprule}{ \lbox (b\imp b)\imp a) \imp( \ldia c\imp \ldia(a\land a))}{
				\vlin{}{\Kdrule}{\ldia c, \lbox ((b\imp b)\imp a)\vdash \ldia(a\land a)}
				{\vliin{}{\limprule}{c , (b \imp b)\imp a \vdash a \land a}
					{\vlin{}{\rimprule}{\vdash b \imp b}{\vlin{}{\AXrule}{b \vdash b}{\vlhy{}}}}
					{\vlin{}{\Wrule}{c,a\vdash a \land a}{\vlin{}{\Crule}{a \vdash a\land a}{
							\vliin{}{\llandrule}{a, a \vdash a\land a}
							{\vlin{}{\AXrule}{a \vdash a}{\vlhy{}}}
							{\vlin{}{\AXrule}{a \vdash a}{\vlhy{}}}
						}}
				}}}
			}
		\end{array}	
		\\[7em]
		\begin{array}{ccccccccccccccc}
		&\vb1 	& \vb0 & \va1 &   &  &&&&& \va0
	 \\[1em]
				&  		&  		&  \vb3 & \vb2 &\va3&&&& &&&\va2 
	 \\[1em]
					&		&			&	     & 	&		&	&&&	 \vldia 0
		\\[1em]	
		\vlbox1	 &		&		&			&	     & 	&&		\vldia 1 {\color{white}c}	
		\\[1.5em]
		\vlbox3	&((\vb5	&\imp	&	\vb4)&\imp &\va7  )	&\imp(& \vldia 3 c &\imp& \vldia 2 & (\va4 &\land &\va6&))
		\end{array}
		\dedges{b1/b0,b3/b2,b0/a1,b2/a3}
		\multimodedges{box1}{a3,a1}
		\multimodedges{dia0}{a0,a2}
		\multidedges{dia1, box1,a1,a3}{a0,a2,dia0}
		\skewedges{box1/box3,dia1/dia3,dia0/dia2,b1/b5,b0/b4,a1/a7,b3/b5,b2/b4,a3/a7,a0/a4,a2/a6}
		\axlbents{b1/b0,b3/b2}
		\axllbents{box1/dia0,dia1/dia0,a1/a0,a3/a2}
		&
		\begin{array}{ccccccccccccccc}
			&\vb1 	&& \vb0 &  \va1 &  & &  &&& \va0
			\\[1em]
			&		&  	 &  &&\va3&&& &&&&\va2 
			\\[1em]
			&			&		&			&	     & 	&		&&&	 \vldia 0
			\\[1em]
			\vlbox1	 &		&		&			&	     & 	&&		\vldia 1 {\color{white}c}	 &	& 
			\\[1.5em]
			\vlbox3	&((\vb5	&\imp	&	\vb4)&\imp &\va7  )	&\imp(& \vldia 3 c &\imp& \vldia 2 & (\va4 &\land &\va6&))
		\end{array}
		\dedges{b1/b0,b0/a3,b0/a1}
		\multimodedges{box1}{a3,a1}
		\multimodedges{dia0}{a0,a2}
		\multidedges{dia1, box1,a1,a3}{a0,a2,dia0}
		\skewedges{box1/box3,dia1/dia3,dia0/dia2,b1/b5,b0/b4,a1/a7,a3/a7,a0/a4,a2/a6}
		\axlbents{b1/b0}
		\axllbents{box1/dia0,dia1/dia0,a1/a0,a3/a2}
	\end{array}
		$
		\end{adjustbox}
		\begin{adjustbox}{max width=\textwidth}
		$
		\small
		\begin{array}{c@{\; , \;}c@{\; , \;}c}
			\begin{array}{|c|c@{}c@{\;}c@{\;}c@{}c@{}c@{}c@{}c@{}c@{}c@{}c@{}c@{}c@{}c@{}c@{}c@{}c|}
				\hline
				&\lbox & (( &b& \imp& b&) \imp &a&) \imp (&\ldia& c&\imp& \ldia&( &a&\land& a &))
				\\\hline
				\evensym &&&&& &&&&& && &&&& \va0 &
				\\\hdashline
				\oddsym &&&&&&&\va1&&&&&&&&&&
				\\\hdashline
				\evensym &&&&&\vb0&&&&&&&&&&&&
				\\\hdashline
				\oddsym	&&&\vb1&&&&&&&&&&&&&&
				\\\hline
			\end{array}
			\dedges{b1/b0,b0/a1,a1/a0}
			\axlbents{b1/b0}
			\axllbents{a1/a0}
			&
			\begin{array}{|c|c@{}c@{\;}c@{\;}c@{}c@{}c@{}c@{}c@{}c@{}c@{}c@{}c@{}c@{}c@{}c@{}c@{}c|}
				\hline
				&\lbox & (( &b& \imp& b&) \imp &a&) \imp (&\ldia& c&\imp& \ldia&( &a&\land& a &))
				\\\hline
				\evensym &&&&&&&&&&&&&&\va0&&&
				\\\hdashline
				\oddsym &&&&&&&\va1&&&&&&&&&&
				\\\hdashline
				\evensym &&&&&\vb0&&&&&&&&&&&&
				\\\hdashline
				\oddsym	&&&\vb1&&&&&&&&&&&&&&
				\\\hline
			\end{array}
			\dedges{b1/b0,b0/a1,a1/a0}
			\axlbents{b1/b0}
			\axllbents{a1/a0}
			&
			\begin{array}{|c|c@{}c@{\;}c@{\;}c@{}c@{}c@{}c@{}c@{}c@{}c@{}c@{}c@{}c@{}c@{}c@{}c@{}c|}
				\hline
				&\lbox & (( &b& \imp& b&) \imp &a&) \imp (&\ldia& c&\imp& \ldia&( &a&\land& a &))
				\\\hline
				\evensym &&&&&&&&&&&& \vldia0&&&&&
				\\\hdashline
				\oddsym		&&&&&&&&&\vldia1&&& &&&&&
				\\\hline
			\end{array}
			\dedges{dia1/dia0}
			\axlbents{dia1/dia0}
		\end{array}
	$
	\end{adjustbox}	

	\caption{
		\textbf{Above:} three derivations of the formula $F=\lbox ((b\imp b)\imp a) \imp( \ldia c\imp \ldia(a\land a))$ together with their corresponding \CKICP. 
		\textbf{Below:} the three maximal views on the modal arena of $F$ in the \CKWIS corresponding to the above proofs.
	}
	\label{fig:intro}
\end{figure*}

\textbf{Intuitionistic combinatorial proofs.}

The syntax of \emph{combinatorial proofs} has been introduced to address the problem of proof equivalence for classical logic  \cite{hughes:pws,hughes:invar}.
In the last years this syntax has been extended to modal logics~\cite{acc:str:CPK}, multiplicative linear logic with exponentials~\cite{acc:EHPN}, relevant logics~\cite{acc:str:wollic19,ral:str:epiccube}, first order logic~\cite{hughes:fopws}, and intuitionistic propositional logic~\cite{ICP}.
	Combinatorial proofs allow to represent ``syntax-free'' proofs, that is, to represent proofs independently from a specific proof system \cite{acc:str:IJCAR18,str:FSCD17}. 
	As consequence, we are able to identify proofs up to some rules permutations, which is the reason why we also refer to combinatorial proof as a semantics for proofs.

In the syntax of \ICPs, formulas are represented by arenas, which are specific directed acyclic graphs,
\begin{equation}\label{eq:firstAr}
	\arof{
		((b\imp b)\imp a )\imp (a\land a)
	}
	=
	\begin{array}{ccccccccc}
		\vb5_1	&	\vb4_0&\va7_1 &\va4_2 &\va6_0
	\end{array}
	\dedges{b5/b4,b4/a7,a7/a4}
	\drbent{a7}{a6}
\end{equation}
and proofs of a formula $F$ are represented as specific graph homomorphisms, called \emph{skew fibrations}, from an arena net (i.e. an arena with an equivalence relation $\axeq$ over vertices satisfying some topological conditions) to the arena of $F$.
In the example below, we represent the $\axeq$-partitions by dashed edges, 
 the skew fibration by dotted arrows, and 
 we write the conclusion $F$ as formula on the left and as arena on the right.
\begin{equation}\label{eq:firstICP}
\hspace{-10pt}
\begin{array}{c@{\;}c@{\;}c@{\;}c@{\;}c@{\;}c@{}c@{\;}c@{\;}c@{}c}
	\vb1 	&  &\vb0&& \va1      && \va0
	\\
	&  		&   &  &&\va3&&&\va2 	
	\\[1.5em]
	((\vb5	&\imp	&\vb4)&\imp &\va7 &)\imp(&\va4 &\land &\va6&)
\end{array}
\dedges{b1/b0,b0/a1}
\multidedges{a1,a3}{a0,a2}
\skewedges{b1/b5,b0/b4,a1/a7,a3/a7,a0/a4,a2/a6}
\axlbents{b1/b0}
\axlbents{a1/a0}
\axrbents{a3/a2}
\hskip2em
\begin{array}{c@{\quad}c@{\quad}c@{\;}c@{\;}c@{\;}c@{\quad}c@{\;\;}c@{\;\;}c}
	\vb1  &\vb0&& \va1    &&& \va0
	\\
	&  		&   &  &&\va3&&\va2 	
	\\[1.5em]
	\vb5&\vb4&& &\va7 &&\va4 &\va6
\end{array}
\dedges{b1/b0,b0/a1}
\multidedges{a1,a3}{a0,a2}
\skewedges{b1/b5,b0/b4,a1/a7,a3/a7,a0/a4,a2/a6}
\axlbents{b1/b0}
\axlbents{a1/a0}
\axrbents{a3/a2}
\dedges{b5/b4,b4/a7,a7/a4}
\drbent{a7}{a6}
\end{equation}
But both represent the same \ICP.

In order to represent proofs of modal formulas, 
we define modal arenas (or \mas), i.e. we characterise labeled \dags with two type of edges which can be uniquely associated to formulas
\begin{equation}
	\arof{(\lbox(b\imp b)\imp \ldia a )\imp \ldia (a\land a)}
	=
	\begin{array}{ccccc}
		&\vlbox1 &\vldia 1& \vldia 0
		\\[1em]
		\vb5	&	\vb4&\va7 &\va4 &\va6
	\end{array}
	\dedges{b5/b4}
	\multidedges{a7,dia1}{a4,dia0}
	\multidedges{b4,box1}{a7,dia1}
	\drbent{a7}{a6}
	\dedges{dia1/a6}
	\modedges{box1/b4,dia1/a7,dia0/a4,dia0/a6}
\end{equation}
We then identify the topological conditions which allows us to represent proofs as skew fibrations from a modal arena net to the \ma  representing the formula to prove.
$$
\begin{array}{c@{}cccccccc@{\;}c@{\;}c@{\;}c@{\;}c@{\;}c@{}c}
	&&\vb1 	&& \vb0 &&  \va1 &  &  &&& \va0
	\\[1em]
	&&		&  	 &  &&&\va3&& &&&&\va2 
	\\[1em]
	\vlbox1&			&&&&&     \vldia 1 &&&	 \vldia 0
	\\[1.5em]
	\vlbox3	&((&\vb5	&\imp	&	\vb4&)\imp &\vldia 3 &\va7  	&)\imp& \vldia 2 & (&\va4 &\land &\va6&)
\end{array}
\dedges{b1/b0,b0/a3,b0/a1}
\multimodedges{box1}{b0}
\multimodedges{dia0}{a0,a2}
\multimodedges{dia1}{a3,a1}
\multidedges{dia1, a1,a3}{a0,a2,dia0}
\skewedges{box1/box3,dia1/dia3,dia0/dia2,b1/b5,b0/b4,a1/a7,a3/a7,a0/a4,a2/a6}
\axlbents{b1/b0}
\axllbents{dia1/dia0,a1/a0,a3/a2}
$$
In particular, each $\axeq$-class (which we represent by linking its vertices by means of dashed edges) in the partition represent either an axiom pairing two atom-labeled vertices, 
or a set of modalities  introduced by a single application of an axiom $\kax$.

\textbf{Game semantics for constructive modal logic.}
In intuitionistic propositional logic, we consider two-players games 
played on the arena of a formula $F$ 
(we denote the players by $\evensym$ and $\oddsym$, but in the literature they are denoted by $\mathsf{O}$ and $\mathsf{P}$, standing for \emph{opponent} and \emph{proponent}).

Each play consists of an alternation  of
\evenmoves and
\oddmoves,
that are vertices of the arena of $F$.
The first move in a play is a \evenmove selected among the $\iedge$-roots of the arena of $F$.
Each subsequent move of a player must be \emph{justified} by a previous move of the other player, that is, the selected vertex must $\iedge$-point a vertex previously played by the other player.
The game terminates when one player has no possible moves, losing the play.

A winning innocent strategy (for $\oddsym$) is a set of plays which takes into account every possible \evenmove, while each \oddmove is uniquely determined (and justified) by one of the previous \evenmoves.

As shown in  \cite{ICP}, the winning strategy over the arena in \Cref{eq:firstAr} with maximal views
$$
\iseven a_0 \isodd a_1 \iseven b_0 \isodd b_1
\qquad\mbox{and}\qquad
\iseven a_2 \isodd a_1 \iseven b_0 \isodd b_1
$$
can be seen as the image of specific paths in the arena net 
by the skew fibration 
in the \ICP in \Cref{eq:firstICP}. Such paths are the ones in which the \evenmove is followed by the unique  \oddmove in the same $\axeq$-class.
In this paper we show that a similar correspondence can be established  for constructive modal logics,
provided some additional conditions on winning strategies for modal arenas
(see \Cref{fig:intro}).

In fact, the presence of modalities requires a new notion of \emph{frames} in a play: 
whenever $\evensym$ plays a move in the scope of a new modality, that is, a modality whose scope contains no previous moves of the play, 
the next \oddmove must be in the scope of the same number of modalities.
This allows us to establish a relation between the modalities in the modal arena and to group them in frames accordingly.
Intuitively, frames allow to certify the correct application of the modal axioms.

\textbf{Outcomes of the paper.}
In this paper we provide the definition of \ICPs and \WISs for the constructive modal logics $\CK$ and $\CD$.
For this purpose, we show a decomposition theorem allowing to transform proofs of a formula 
into factorised proofs, that are proofs consisting of a linear part and a weakening-contraction part.
We show soundness and completeness of these semantics and we prove the following full completeness result:
$$
\Set{\mbox{factorised proof of } F}
\twoheadrightarrow
\Set{\mbox{\ICPs of }F}
\twoheadrightarrow
\Set{\mbox{\WISs on } \arof F}
$$
To our knowledge no game semantics for modal logic are discussed in the literature. Our game semantics approach paves the way to  concrete full-complete denotational models for modal logics.

\textbf{Organisation of the paper.}
In \Cref{sec:background} we show a decomposition theorem by providing a polarized sequent calculus~\cite{lamarche:essential,marin:str:aiml14} 
which also include some deep inference rules~\cite{gug:str:01,brunnler:tiu:01,guglielmi:SIS};
in \Cref{sec:arenas} we establish a correspondence between certain labeled directed graphs (modal arenas) and modal formulas;
these graphs, enriched with a partition of their vertices, are used in \Cref{sec:arenaNets} to encode linear proofs;
moreover, in \Cref{sec:skew} we show how to represent structural derivations by means of skew fibrations between modal arenas;
in \Cref{sec:CP} we provide a definition of \ICPs for $\CK$ and $\CD$, and in \Cref{sec:games} we use them to define winning innocent strategies for these logics.

\begin{figure*}[!t]
	\centering
	\def\myskip{\hskip.3em}
			\begin{tabular}{c@{\myskip}c@{\myskip}c@{\myskip}c@{\myskip}c@{\myskip}c@{\myskip}c@{\myskip}c}
				$\vlinf{}{\AXrule}{a \vdash a}{}$
				&
				$\vliinf{}{\cutr}{ \Gamma, \Delta \vdash B}{ \Gamma \vdash A}{ \Delta, A \vdash B}$
				&
				$\vlinf{}{\Crule}{ \Gamma, A \vdash B}{ \Gamma, A, A \vdash B}$
				&
				$\vlinf{}{\Wrule}{ \Gamma, A \vdash B}{ \Gamma \vdash B}$
				\\[1em]
				$\vlinf{}{\rimprule}{ \Gamma  \vdash A \imp B}{ \Gamma , A \vdash B}$
				&
				$\vliinf{}{\limprule}{ \Gamma, \Delta , A \imp B\vdash C}{ \Gamma \vdash A}{ \Delta, B \vdash C}$
				&
				$\vliinf{}{\rlandrule}{ \Gamma, \Delta \vdash A \land B}{ \Gamma \vdash A}{ \Delta \vdash B}$
				& 
				$\vlinf{}{\llandrule}{ \Gamma, A \land B \vdash C}{ \Gamma, A, B \vdash C}$
				\\[1em]
				$\vlinf{}{\Kbrule}{ \lbox \Gamma \vdash \lbox A }{ \Gamma \vdash A}$
				&
				$\vlinf{}{\Kdrule}{\ldia A, \lbox \Gamma \vdash \ldia B}{A, \Gamma \vdash B}$
				&
				$
				\vlinf{}{\Kurule}{ \ldia \unit, \lbox \Gamma \vdash \ldia A }{ \Gamma \vdash A}
				$
				&
				$\vlinf{}{\Drule}{\lbox \Gamma \vdash \ldia A}{\Gamma \vdash A}$
			\end{tabular}
			
			$$\begin{array}{rcl@{\qquad}rcl}
				\IMLL& =& \set{\AXrule, \rimprule, \limprule,\llandrule, \rlandrule }
				&
				\Ip&=& \IMLL\cup\set{\Crule,\Wrule}
				\\
				\IMLLK& =& \IMLL\cup\set{\Kbrule,\Kdrule,\Kurule}
				&
				\IpK &=& \Ip\cup \set{\Kbrule,\Kdrule
				}
				\\
				\IMLLD& =& \IMLL\cup\set{\Kbrule,\Kdrule,\Drule}
				&
				\IpD&= &\Ip \cup \set{\Kbrule,\Kdrule,\Drule}
			\end{array}
			$$
	\caption{Sequent rules and sequent systems for constructive modal logics considered in this paper}
	\label{fig:seqCalc}
\end{figure*}

\section{Preliminaries on Constructive modal logics}\label{sec:background}

\begin{figure*}[!t]
	\def\myskip{\hskip.5em}
			\begin{tabular}{c@{\myskip}c@{\myskip}c@{\myskip}c@{\myskip}c}
				$\vlinf{}{\axrule}{ \isodd{a}, \iseven{a}}{}$
				&
				&
				$\vlinf{}{\worule}{ \isodd \Gamma, \isodd B, \iseven A}{ \isodd \Gamma , \iseven A}$
				&
				$\vlinf{}{\crule}{ \isodd \Gamma,\isodd  A \iseven B}{ \isodd \Gamma,\isodd  A, \isodd A ,\iseven B}$
				\\[1em]
				$\vlinf{}{\rlimpr}{\isodd \Gamma , \iseven{( A \limp  B)}}{ \isodd \Gamma , \isodd A ,\iseven B}$
				&
				$\vliinf{}{\llimpr}{ \isodd \Gamma , \isodd \Delta,  \isodd{( A \limp  B)} , \iseven C}{ \isodd \Gamma ,\iseven A}{ \isodd\Delta, \isodd B ,\iseven C}$
				&
				$\vliinf{}{\rltensr}{ \isodd \Gamma, \isodd \Delta ,\iseven{( A \ltens  B)}}{ \isodd\Gamma ,\iseven A}{ \isodd \Delta, \iseven B}$
				&
				$\vlinf{}{\lltensr}{ \isodd \Gamma, \isodd{( A \ltens  B)} , \iseven C}{ \isodd \Gamma, \isodd A, \isodd B , \iseven C}$
				\\[1em]
				$\vlinf{}{\kbrule}{ \isodd {\lbox \Gamma},\iseven {\lbox A} }{ \isodd \Gamma ,\iseven A}$
				&
				$\vlinf{}{\kdrule}{\isodd {\ldia A}, \isodd {\lbox \Gamma} , \iseven {\ldia B}}{\isodd A, \isodd \Gamma ,\iseven{B}}$
				&
				$\vlinf{}{\kurule}{\isodd {\ldia \unit}, \isodd {\lbox \Gamma} , \iseven {\ldia B}}{ \isodd \Gamma ,\iseven{B}}$
				&
				$\vlinf{}{\drule}{\isodd {\lbox \Gamma} , \iseven {\ldia A}}{ \isodd \Gamma ,\iseven{A}}$
			\end{tabular}
			$$\begin{array}{r@{\;}c@{\;}l@{\qquad}r@{\;}c@{\;}l}
				\PIMLL & =&  \set{\axrule,\rlimpr , \llimpr, \lltensr,\rltensr}     
				&
				\PIp&=& \PIMLL \cup \set{\crule,\worule} 
				\\
				\PIMLLK &=& \PIMLL \cup \set{\kbrule, \kdrule,\kurule}
				&
				\PK &=& \PIp\cup \set{\kbrule,\kdrule
				}
				\\
				\PIMLLD&= &\PIMLL \cup \set{\kbrule,\kdrule,\drule}
				&
				\PD&= & \PIp \cup \set{\kbrule, \kdrule,\drule}
			\end{array}
			$$

	        \caption{Polarised sequent rules and polarised sequent systems for the constructive modal logics considered in this paper}
	\label{fig:PolSeqCalc}\label{fig:paritySystems}
\end{figure*}

\begin{figure*}[!t]
	\def\myskip{\hskip1em}
	\hbox to\textwidth{\hfil   
		$
			\DOWN{\PIp}
			\;
			=
			\Set{
			\;
			\vlinf{}{\deep\wdrule}{\Gamma\cons {\isodd{\ldia A}}}{ \Gamma\cons{ \isodd{\ldia \unit} }}
			\; , \;
			\vlinf{}{\deep\wtrule}{\Gamma\cons {\isodd{ A \ltens  B}}}{ \Gamma\cons{ \isodd A }}
			\;    , \;
			\vlinf{}{\deep\wirule}{\Gamma\cons {\iseven{B\limp A}}}{ \Gamma\cons{ \iseven A }}
			\;    , \;
			\vlinf{}{\deep \crule}{\Gamma\cons {\isodd A}}{ \Gamma\cons {\isodd{ (A \ltens  A)}}}
			\;
			}
		$
		\hfil}
	\caption{Deep inference rules for weakening and contraction}
	\label{fig:deepRules}
\end{figure*}

In this paper we consider the \emph{modal formulas} generated by a countable set of (atomic) propositional variables $\atoms= \set{a, b, \dots}$ via the following grammar
$$A,B ::=\  a \mid   A\imp B \mid A\land B \mid \lbox A \mid \ldia A
\mid \ldia \unit
$$
and we say that a formula is \emph{modality-free} if it contains no occurrences of $\lbox $  and $\ldia$.

We consider the variant of intuitionistic modal logic $\CK$ called \emph{constructive modal logic} \cite{ari:das:str:constructive, bie:paiva:intuitionistic, mend:multimodal, wij:constr, kuz:mar:str:Justification} defined by adding to the intuitionistic propositional logic 
the necessitation rule
$$
\mbox{If $F$ is provable, then $\lbox F$ is provable}
$$
and the two following axiom schemes $\krule_1$ and $\krule_2$.
\begin{equation*}\label{eq:Kaxioms}
\!	\krule_1  \colon  \lbox(A \imp B) \imp (\lbox A \imp \lbox B)  
	\qquad
	\krule_2    \colon        \lbox(A \imp B) \imp (\ldia A \imp \ldia B)      
\end{equation*}

A further extension of this logic, denoted $\CD$, can be obtained by adding the following axiom scheme
$$
	\drule     \colon        \lbox A \imp \ldia A
$$

In this paper we consider the fragment of $\CK$ and $\CD$ containing only implication and conjunction given by the rules and sytems in \Cref{fig:seqCalc}, since these suffice to express $\lambda$-calculi with pairs for these logics.

We remark that the presence of $\ldia \unit$ is not standard for the unit-free fragment.
In fact, no rule can introduce this formula in $\IpK$ and $\IpD$.
However, it plays a special role in some results in this paper since its purpose is to represent a ``placeholder'' for a $\ldia$-formula which may be introduced by a weakening rule but whose occurrence is not a negligible information in the proof.

\begin{theorem}\label{thm:soundCompleteLogic}
	The sequent system $\IpX$ is a sound and complete proof system for the disjunction-free fragment of the logic $\X$ for all $\X \in \set{\CK,\CD}$.
	Moreover, $\IpX$ satisfies cut-elimination property.
\end{theorem}
\begin{proof}
	In \cite{kuz:mar:str:Justification} there are provided sound and complete systems for these logics.
	These systems are proven to be analytic, i.e. satisfying cut-elimination property.
	Thus we can extract the desired disjunction-free calculi.
\end{proof}

In order to define combinatorial proofs for a given logic, we need to have a \emph{decomposition theorem} which lets us factorize proofs in a linear part, capturing the logic interactions between the components of the proof, and a resource management part, capturing resources duplication or erasing.

To achieve this decomposition result for the logics considered in this paper, 
we use the sound and complete cut-free sequent systems provided in \cite{marin:str:aiml14} to define new rule systems in which we make use of deep inference rules~\cite{gug:gun:par:2010,gug:str:01,brunnler:tiu:01,guglielmi:SIS}, that is,  rules which can be applied deep inside a formula in any context.
The use of deep inference rules allows us to push down in a derivation all the occurrences of weakening and contractions.
In particular, as done in \cite{acc:str:CPK} for classical modal logic, we consider $\Krule$ and $\Drule$ as part of the logic interaction of a proof. 

\subsubsection*{Polarized formulas}

However, a difficulty arises in applying such permutations in the intuitionistic setting since weakening and contraction rules may be performed only on the left-hand-side formulas in a sequent.
In order to assure the correctness of deep applications of  weakening and contraction rules, 
we introduce a syntax using \emph{polarized formulas} (or \emph{\pformulas}) to represent a two-sided single-conclusion calculus by a polarized one-sided sequent calculus, as done in \cite{marin:str:aiml14}. This allows us to  keep track of which subformulas in a sequent $\Gamma$ occurred on the left-hand-side of a sequent occurring in a derivation of $\Gamma$.

We define the set of \pformulas as the set generated by $\atoms = \set{a, b, \dots}$ using the following grammar
$$
\begin{array}{c@{::=}l}
	\iseven A, \iseven B
	&
	\iseven a \mid \iseven A \ltens \iseven B \mid \isodd A \limp \iseven B \mid \lbox \iseven A \mid \ldia \iseven A
	\\
	\isodd A, \isodd B
	&
	\isodd a \mid \isodd A \ltens \isodd B \mid \iseven A \limp \isodd B \mid \lbox \isodd A \mid \ldia \isodd A \mid \ldia \isodd \unit
\end{array}
$$
A \emph{context} is a sequent $\Gamma{\cons{}}$ in which one atom occurrence is been replaced by the hole~$\cons{}$. 

In order to improve readability, we omit to write polarities on subformulas since they can be deduced as follows:
\begin{itemize}
	\item if $\iseven{(A \limp B)}$, then $\isodd A$ and $ \iseven B$;
	\item if $\isodd{(A \limp B)}$, then $\iseven A$ and $ \isodd B$;
	\item if $\iseven{(A \ltens B)}$, then $\iseven A$ and $ \iseven B$;
	\item if $\isodd{(A \ltens B)}$, then $\isodd A$ and $ \isodd B$;
	\item if $\iseven{\lbox A}$ or $\iseven{\ldia A}$, then $\iseven A$; 
	\item if $\isodd{\lbox A}$ or  $\isodd{\ldia A}$, then $\isodd A$;
\end{itemize}

\begin{figure*}
	\centering
	\small
        \def\myskip{-1.5ex}
	$
	\begin{array}{c@{\hskip.5em}c@{\hskip.5em}c}
		\vlderivation{
			\vlin{}{\kbrule}{\isodd{\lbox \Gamma}, \isodd{\lbox A} ,\iseven{\lbox B}}{
				\vlin{}{\worule}{\isodd \Gamma, \isodd A ,\iseven B}
				{\vlhy{\isodd \Gamma, \iseven B}}
		}}
		\normto
		\vlderivation{
			\vlin{}{\worule}{\isodd{\lbox \Gamma}, \isodd{\lbox A} ,\iseven{\lbox B}}{
				\vlin{}{\kbrule}{\isodd {\lbox\Gamma}, \iseven{\lbox B}}
				{\vlhy{\isodd \Gamma, \iseven B}}
		}}
		&
		\vlderivation{
			\vlin{}{\kdrule}{\isodd{\lbox \Gamma}, \isodd{\ldia A} ,\iseven{\ldia B}}{
				\vlin{}{\worule}{\isodd \Gamma, \isodd A , \iseven B}
				{\vlhy{\isodd \Gamma, \iseven B}}
		}}
		\normto
		\vlderivation{
			\vlin{}{\deep\wdrule}{\isodd{\lbox \Gamma}, \isodd{\ldia A} ,\iseven{\ldia B}}{
				\vlin{}{\kurule}{\isodd {\lbox\Gamma},\isodd{\ldia \unit}, \iseven{\ldia B}}
				{\vlhy{\isodd \Gamma, \iseven B}}
		}}
		%
		\\ \\[\myskip]
		\vlderivation{
			\vlin{}{\drule}{\isodd{\lbox \Gamma}, \isodd{\lbox A} ,\iseven{\ldia B}}{
				\vlin{}{\worule}{\isodd \Gamma, \isodd A , \iseven B}
				{\vlhy{\isodd \Gamma, \iseven B}}
		}}
		\normto
		\vlderivation{
			\vlin{}{\worule}{\isodd{\lbox \Gamma}, \isodd{\lbox A} ,\iseven{\ldia B}}{
				\vlin{}{\drule}{\isodd {\lbox\Gamma},\iseven{\ldia B}}
				{\vlhy{\isodd \Gamma, \iseven B}}
		}}
		%
		&
		\vlderivation{\vlin{}{\lltensr}{ \isodd \Gamma, \isodd{( A \ltens  B)} , \iseven C}{\vlin{}{\worule} {\isodd \Gamma, \isodd A, \isodd B , \iseven C}{\vlhy{ \isodd \Gamma, \isodd A, \iseven C}}}}
		\normto
		\vlderivation{\vlin{}{\deep\wtrule}{ \isodd \Gamma, \isodd{( A \ltens  B)} , \iseven C}{{\vlhy{ \isodd \Gamma, \isodd A, \iseven C}}}}
		%
		\\\\[\myskip]
		\vlderivation{\vlin{}{\rlimpr}{ \isodd \Gamma, \iseven{B\limp A}}{\vlin{}{\worule} {\isodd \Gamma,  \isodd B , \iseven A}{\vlhy{ \isodd \Gamma, \iseven A}}}}
		\normto
		\vlderivation{\vlin{}{\deep\wirule}{ \isodd \Gamma, \iseven{  B\limp A } }{{\vlhy{ \isodd \Gamma, \iseven A}}}}
                &
		\vlderivation{
			\vliin{}{\llimpr}{\Gamma, A \limp B, \Delta \vdash C}
			{\vlpr{\dD}{\PX}{\isodd \Gamma, \iseven  A }}
			{\vlin{}{\worule}{\isodd B,\isodd \Delta , \iseven C}
				{\vlhy{\isodd \Delta, \iseven C}}
		}}
		\quad
		\normto
		\quad
		\vlderivation{
			\vliq{}{\worule}{\isodd \Gamma, \isodd{A \limp B}, \isodd \Delta, \iseven C}
			{\vlhy{\isodd \Delta, \iseven C}}
		}
	\end{array}
	$
	\caption{Rule permutations for $\worule$}
	\label{fig:permutations}
\end{figure*}

For \pformulas, we define the sequent rules as systems in \Cref{fig:PolSeqCalc}, and deep rules in \Cref{fig:deepRules}.
In particular, the rules  in \Cref{fig:PolSeqCalc}
can be obtained by the ones in \Cref{fig:seqCalc} by encoding any two sided sequent $\Gamma, B \vdash A$ as the one-side sequent $\isodd \Gamma, \isodd B, \iseven A$.

Polarized formulas allow us to restrain the application of the deep rules only to specific subformulas. 
In particular, we can apply  $\crule, \wdrule, \worule$
($\wirule$) to the formulas which occurs as a $\oddsym$-formula (respectively $\evensym$-formula) in a sequent occurring in the derivation.

If $H$ is a \pformula, we denote by $\spalletto H$ the formula obtained by removing all polarities occurring in $H$ and replacing the $\ltens$ and $\limp$ symbols respectively with  $\land$ and $\imp$.
This translation induces a correspondence between the systems in \Cref{fig:seqCalc} and in \Cref{fig:PolSeqCalc}.
However, the interest in introducing the polarized systems depends on the following result.

\begin{notation}
	If $S$ is a set of rules, we write $F'\provevia S F$ if there is a derivation from $F'$ to $F$ using rules in $S$.
	Moreover, we write $\provevia S F$  if there is a proof of $F$ in $S$, i.e. a derivation using rules in $S$ form the empty premise.
\end{notation}

\begin{theorem}\label{thm:pol}
	Let $\X \in \set{\CK, \CD}$ and  $H$ be a \pformula, then 
	\begin{itemize}
		\item $\provevia{\IpX} \spalletto H$ iff $\provevia{\PX} H$;
		\item $\provevia{\IMLLX} \spalletto H$ iff $\provevia{\PIMLLX} H$.
	\end{itemize}
\end{theorem}

\begin{theorem}[Decomposition]\label{thm:decomposition}
	Let $\X \in \set{\K, \D}$ and $H$ be a \pformula. Then 
	$\provevia{\PX} H $
	iff
	$ \provevia{\PIMLLX} H' \provevia{\DOWN{\PIp}} H$.
\end{theorem}
\begin{proof}
	Let us consider an $\PX$-derivation of $H$.
	We replace every occurrence of a $\crule$ by a $\lltensr$ followed by a $\deep \crule$.
	We then permute every occurrence of $\worule$ and rules in $\DOWN{\PIp}$ by applying independent rule permutations 
	plus 
	the permutations in \Cref{fig:permutations} replacing a $\worule$ into a  $\deep\wtrule$, $\deep \wirule$ or $\deep \wdrule$.
	Observe since $H$ is a \pformula, then no $\worule$-rule occurs in the derivation at the end of this procedure.
	We conclude by permuting all deep-weakening and deep-contraction rules down in the derivation.
	
	During this process, depending on the presence of $\kurule$ or $\drule$ in $\CX$, an occurrence of a $\kdrule$ may be replaced in two different ways: either by $\kurule$ followed by a $\deep \wdrule$ 
	or
	by a $\drule$ followed by a $\worule$ 
	as shown below:
	{\small{
			$$
			\vlderivation{
				\vlin{}{\kdrule}{\isodd{\ldia A} , \isodd {\lbox B} , \iseven{ \ldia C}}{
					\vlin{}{\worule}{\isodd A, \isodd B, \iseven C}
					{\vlpr{}{}{\isodd B, \iseven  C}}
			}}
			\normto
			\quad
			\vlderivation{
				\vlin{}{\deep\wdrule}{\isodd{\ldia A }, \isodd{\lbox B}, \iseven {\ldia C}}{
					\vlin{}{\kurule}{\isodd{\ldia \unit}, \isodd{\lbox B},\iseven{\ldia C}}
					{\vlpr{}{}{\isodd B, \iseven C}}
			}}
			\;\mbox{ or }\;
			\vlderivation{
				\vlin{}{\worule}{\ldia \isodd A , \lbox \isodd B , \ldia \iseven C}{
					\vlin{}{\drule}{\lbox \isodd B, \ldia \iseven C}
					{\vlpr{}{}{\isodd B, \iseven C}}
			}}
			$$
	}}
	To prove the converse it suffice to revert the previous procedure. 
\end{proof}

	If $F$ is a formula, we call a \emph{factorised proof} of $F$ a derivation in 
	$\PIMLLX\cup\DOWN{\PIp}$ of the form $ \provevia{\PIMLLX} H' \provevia{\DOWN{\PIp}} H$,
	for $H$ and $H'$ \pformulas such that $F=\spalletto H$.

\section{Modal arenas}\label{sec:arenas}

In this section we establish a correspondence between modal formulas and a family of labeled directed graphs we call \emph{modal arenas}.
These are employed in this paper in the definition of intuitionistic combinatorial proofs and games.

A \emph{directed graph} $\gG =\tuple {\vertices[\gG], \diredge[\gG]}$ is given by a set of vertices $\vertices[\gG]$ and a set of direct edges $\diredge[\gG]\subseteq \vertices[\gG]\times \vertices[\gG]$. 
If $V' \subset \vertices[\gG]$, we say that $\tuple{V' , \diredge\cap(V' \times V'  )}$ is the \emph{subgraph of $\gG$ induced by $V'$}.
We write $u\diredge[\gG] v$, $u\not \!\!\!\diredge[\gG] v$, $u\codiredge[\gG] v$ and $u\sdiredge[\gG] v$ if respectively $uv\in \diredge[\gG]$, $uv\not\in \diredge[\gG]$, $vu\in \diredge[\gG]$ and $uv\in\diredge[\gG]\cup\codiredge[\gG]$.
A vertex $v$ is a \emph{$\diredge[\gG]$-root}, denoted $\isdirroot v$ 
if there is no vertex $w$ such that $v\diredge[\gG] w$.
We call  $\dirrof[\gG]=\set{v\mid \isdirroot v}$ the set of $\diredge[\gG]$-roots of $\gG$.

A \emph{path} from $v$ to $w$ of length $n$ is a sequence  of vertices $x_0\dots x_n$ such that $v=x_0$, $w=x_n$ and  $x_i\diredge[\gG] x_{i+1}$ for $i\in \set{0, \dots, n-1}$. We write $v\diredge[\gG]^{*}w$ ($v\diredge[\gG]^{n}w$) if there is a path (respectively a path of length $n$) from $v$ to $w$.
A \emph{directed acyclic graph} (or \dag for short) is a direct graph such that $v\diredge[\gG]^{n}v$ implies $n=0$ for all $v\in \vertices$.

A \emph{two-color directed acyclic graph}
(or \ddag for short) $\gG=\tuple {\vertices[\gG], \iedge[\gG],\medge[\gG]}$ is given by a set of vertices $\vertices[\gG]$ and two disjoint sets of edges $\iedge[\gG]$ and $\medge[\gG]$ such that the graph $\tuple {\vertices[\gG],\iedge[\gG]\cup \medge[\gG]}$ is acyclic.
We denote $\diedge[\gG]=\iedge[\gG]\cup \iiedge[\gG]$ and  $\dmedge[\gG]=\medge[\gG]\cup\imedge[\gG]$.
We omit the superscript when clear from context.

If $\labelset$ is a set, a \ddag is \emph{$\labelset$-labeled} if a \emph{label}  $\lab v\in \labelset$ is associated to each vertex $v\in V$.
In this paper we fix the set of labels to be the set $\labelset =\atoms \cup \set{\lbox, \ldia}$, where $\atoms$ is the set of propositional variables occurring in formulas.

\newcommand{\RtoR}[2]{R^#1_#2}

\begin{definition}
	Let $\gG$ and $\gH$ be \ddags, we denote by $\RtoR{\gG}{\gH}$ the set of edges from the $\iedge$-roots of $\gG$ to the $\iedge$-roots of $\gH$, that is $\RtoR{\gG}{\gH} = \set{(u,v)\mid u\in \irof[\gG], v\in \irof[\gH]}$.
	
	We define the following operations on \ddags:
	$$
	\begin{array}{r@{=}c@{\;}l@{\:,\:}l@{\:,\:}l@{\;}r}
		\gG \gsum \gH 
		& 
		\tuple{
			&
			\vertices[\gG]\cup \vertices[\gH]
			& 
			\iedge[\gG]\cup \iedge[\gH]
			& 
			\medge[\gG]\cup \medge[\gH]
			&
		}
		\\
		\gG \gimp \gH
		&
		\tuple{
			&
			\vertices[\gG]\cup \vertices[\gH]
			& 
			\iedge[\gG]\cup \iedge[\gH] \cup
			\RtoR\gG\gH
			&
			\medge[\gG]\cup \medge[\gH]
			&
		}
		\\
		\gG \gmod \gH
		&
		\tuple{
			&
			\vertices[\gG]\cup \vertices[\gH]
			&
			\iedge[\gG]\cup \iedge[\gH]
			&
			\medge[\gG]\cup \medge[\gH]\cup
			\RtoR\gG\gH
			&
		}
	\end{array}
	$$
	which can be pictured as follows, with $\rootsym$ representing the \hbox{$\iedge$-roots} of each graph.
	$$
	\begin{array}{c@{\;}|@{\quad}c@{\;}|@{\quad}c}
		\gG\gsum\gH &\gG\gimp\gH &\gG\gmod\gH \\
		\hline
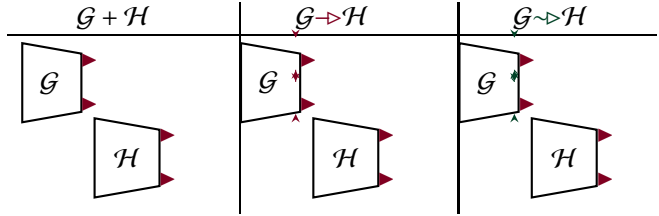
		\begin{tikzpicture}
			\node (A) at (0,0) [trapezium, trapezium angle=80, minimum width=30pt, draw, thick,rotate=-90] {\rotatebox{90}{$\gG\,$}};
			\node at (0.55,0) {$\begin{array}{c}\vrroot1 \\[5pt] \vrroot2\end{array}$};
			\node (A) at (1,-1) [trapezium, trapezium angle=80, minimum width=30pt, draw, thick,rotate=-90] {\rotatebox{90}{$\gH$}};
			\node at (1.59,-1) {$\begin{array}{c}\vlroot3 \\[4.5pt]  \vlroot4\end{array}$};
		\end{tikzpicture}
		&
		\begin{tikzpicture}
			\node (A) at (0,0) [trapezium, trapezium angle=80, minimum width=30pt, draw, thick,rotate=-90] {\rotatebox{90}{$\gG\,$}};
			\node at (0.55,0) {$\begin{array}{c}\vrroot1 \\[5pt] \vrroot2\end{array}$};
			\node (A) at (1,-1) [trapezium, trapezium angle=80, minimum width=30pt, draw, thick,rotate=-90] {\rotatebox{90}{$\gH$}};
			\node at (1.59,-1) {$\begin{array}{c}\vlroot3 \\[4.5pt]  \vlroot4\end{array}$};
		\end{tikzpicture}
		\multidedges{m1,m2}{m4,m3}
		&
		\begin{tikzpicture}
			\node (A) at (0,0) [trapezium, trapezium angle=80, minimum width=30pt, draw, thick,rotate=-90] {\rotatebox{90}{$\gG\,$}};
			\node at (0.55,0) {$\begin{array}{c}\vrroot1 \\[5pt] \vrroot2\end{array}$};
			\node (A) at (1,-1) [trapezium, trapezium angle=80, minimum width=30pt, draw, thick,rotate=-90] {\rotatebox{90}{$\gH$}};
			\node at (1.59,-1) {$\begin{array}{c}\vlroot3 \\[4.5pt]  \vlroot4\end{array}$};
		\end{tikzpicture}
		\multimodedges{m1,m2}{m4,m3}
	\end{array}
	$$
\end{definition}

We use the notation $\singlevertex a$, $\ssbox$ and $\ssdia$ for the graph consisting of a single vertex labeled respectively by $a$, $\lbox$ and $\ldia$.
If $F$ is a formula, we define a $\labelset$-labeled \ddag $\arof F$ inductively as follows:
\begin{equation}
	\label{eq:translation}
	\begin{array}{c@{\;\;=\;\;}l}
		\arof a & \singlevertex[a]\\
		\arof{A \imp B}& \arof A \gimp \arof B\\
		\arof{A\land B} & \arof A\gsum \arof B 
	\end{array}
	\hskip1.5em
	\begin{array}{c@{\;\;=\;\;}l}
		\arof{\ldia \unit} & \ssdia \\
		\arof{\lbox A}& \, {\ssbox}\;  \gmod \arof A\\
		\arof{\ldia A} & {\ssdia}\; \gmod \arof A
	\end{array}
\end{equation}
Using the same notation, if $H$ is a \pformula and $F$ the formula such that $F=\spalletto H$, we denote by $\arof H$ the \ddag $\arof{F}$.

In order to characterize those \ddags that are encoding of formulas, we require some additional definitions.

\begin{definition}\label{def:pa}
	
	A $\labelset$-labeled \dag $\gG=\tuple {\vertices[\gG], \iedge[\gG]}$ is a \emph{arena} if 
	$\vertices[\gG]\neq \emptyset$ 
	and if it is
	\begin{itemize}
		\item L-free: if $a \iedge u$ and $a \iedge w \iedge v$ then $u\iedge v$;
		\item $\Sigma$-free: if $a\iedge v$, $a \iedge w$, $b\iedge w$ and $b\iedge u$ then $a\iedge u$ or $b\iedge v$;
	\end{itemize}
	That is, the following induced subgraphs are forbidden.
	
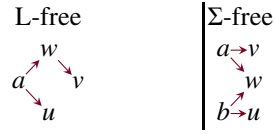
\begin{center}
		\begin{tabular}{c@{\qquad}|@{\qquad}c}
			L-free &
			$\Sigma$-free 
			\\
			{$
				\begin{array}{c@{\;\;}c@{\;\;}c}
					&\vw1\\
					\va1&&\vv1\\
					&\vu1 
				\end{array}
				\dedges{w1/v1,a1/w1,a1/u1}
				$}
			&
			{$
				\begin{array}{c@{\;\;}c@{\;\;}c}
					\va1 & \vv1
					\\
					&\vw1
					\\
					\vb1 & \vu1
				\end{array}
				\dedges{a1/v1,a1/w1,b1/w1,b1/u1}
				$}
		\end{tabular}
	\end{center}
\end{definition}

We recall some results from~\cite{ICP} on arenas and modality-free formulas.

\begin{lemma}[\cite{ICP}]\label{lemma:cones}
	In an arena, if $v\iedge^n y$ and $w\iedge^m y$, then
	$\icone v n \subseteq \icone w m$ or $\icone w m \subseteq \icone v n$ 
	where $\icone u k=\set{x\mid u\iedge^k x}$.
\end{lemma}

\begin{theorem}[\cite{ICP}]\label{thm:prearena}
	A graph $\gG$ is an arena iff there is a modality-free formula $F$ such that $\gG=\arof{F}$.
\end{theorem}

\begin{definition}\label{def:ma}
  	A \emph{modal arena} (or \ma) $\gG = \tuple {\vertices[\gG], \iedge[\gG],\medge[\gG]}$ is an $\labelset$-labeled \ddag such that

  \begin{itemize}

  	\item  $\tuple {\vertices, \iedge}$ is an arena;
  	
  	\item $\medge $ is modal, that is:
  	\begin{itemize}
  		\item		if $v\medge w$ and $w\medge u$, then $v\medge u$ (transitivity);
  		\item		if $v\medge w$ and $u\medge w$, then $u \dmedge v$;
  		\item		if $v\medge w$ and $v\medge u$, then $u \niedge w$;	
  		\item		if $v\medge w$ and $u\iedge v$, then $u\iedge w$;
  		\item		if $v\medge w$ and $v\iedge u$,  then $w\iedge u$;
  		\item		if $v\medge w$ and $w\iedge u$,  then $v\iedge u$;
  	\end{itemize}

  	\item $\gG$ is \emph{properly labeled}:
  	\begin{itemize}
  		\item   if $v\medge w$, then $\lab v\in \set{\lbox, \ldia}$;
  		\item	if $\lab v=\lbox$, then there is a $w$ such that $v\medge w$.
  	\end{itemize}
  \end{itemize}

  We denote by $\avertices[\gG]$, $\bvertices[\gG]$ and $\dvertices[\gG]$ the subsets of vertices of $\gG$ with labels respectively in $\atoms$, $\set\lbox$ and $\set\ldia$.
  We call \emph{atomic} the vertices in $\avertices[\gG]$ and \emph{modal} the ones in $\mvertices[\gG]=\bvertices[\gG] \cup \dvertices[\gG]$.
\end{definition}

From Lemma~\ref{lemma:cones} we can prove the following:
\begin{lemma}\label{lemma:modalities}
	Let $\gG$ be a \ma and $u,v,w\in \vertices[\gG]$. If $v\medge w$ then:
	\begin{itemize}
		\item 	$v$ is a $\iedge$-root  iff $w$ is a $\iedge$-root;
		\item	$v\iedge^n u$ iff $w\iedge^n u$;
		\item 	if $u\iedge^n v$ then $u\iedge^n w$.
	\end{itemize}
\end{lemma}
\begin{proof}
	The first statement follows the fact that in a \ma if $v\medge w$, then $v\iedge u$ iff $w\iedge u$.
	The second statement is proven using the same argument, proceeding by induction on $n$ making use of  Lemma~\ref{lemma:cones}.
	The third statement is also proven using Lemma~\ref{lemma:cones} and the fact that in a \ma if $v\medge w$ and $u\iedge v$, then $u\iedge w$.
\end{proof}

\begin{lemma}\label{thm:forToAr}
	If  $F$ is a formula, then the $\labelset$-labeled \ddag $\arof F$ is an \ma.
\end{lemma}
\begin{proof}
	The right-to-left implication is proven by induction over the number of connectives and modalities of a formula.
	It suffices to remark that the graph operations $\gsum$ and $\gimp$ cannot introduce forbidden {\ma} configurations.
	Similarly, the operation $\gmod$ introduces no forbidden configurations whenever $\gG=\gG_1\gmod \gG_2$ with $\gG_1$ a single vertex graph of the form $\ssbox$ or~$\ssdia$.
\end{proof}

For proving the converse, we need the following concept.
If $v\in \mvertices$ is a vertex in a \ma,  then we define the \emph{$\medge$-cone of $v$} as the set of vertices 
$$\mconeof v = \set{w\mid  \mbox{there is } u \mbox{ such that } v\medge u \mbox{, } w\iedge^{*}u \mbox{ and } 	w\niedge^{*}v}$$
Intuitively, the cone of a modal vertex delimits the subformula in the scope of the corresponding modality.

\begin{example}\label{ex:cones}
	Consider the formula $F=\big(a\imp\lbox (b\land (c\imp \ldia d))\big)\imp \ldia (e\imp f)$ and its \ma
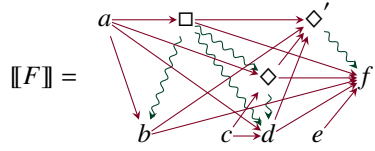
	$$
	\arof{F}
	=
	\begin{array}{ccccccc}
		\va1		&			&\vlbox1&		&& \vldia2'
		\\[1em]
		&		&&&\vldia1	& 	&\vf1
		\\[1em]
		& \vb1&	  &	\vc1	& \vd1	&\ve1	
	\end{array}
	\modedges{box1/dia1,box1/d1,box1/b1,dia2/f1,dia1/d1}
	\multidedges{c1}{d1,dia1}
	\multidedges{a1}{b1,dia1,d1,box1}
	\multidedges{b1,dia1,d1,box1}{dia2,f1}
	\dedges{e1/f1}
	$$
	The $\lbox$ modality has subformula  $b\land (c\imp \ldia d)$, the first $\ldia$ has subformula $d$ and the second $\ldia $ (denoted $\ldia'$ on the graph) has subformula $e\imp f$.
	The corresponding $\medge$-cones are $\mconeof{\lbox}=\set{b,c,\ldia, d}$, $\mconeof{\ldia}=\set d$ and $\mconeof{\ldia'}=\set{e,f}$.
\end{example}

If $v$ is a vertex of a \ma $\gG$, we call the \emph{principal modal vertex of $v$} the unique\footnote{Its uniqueness follows by definition of \ma.} vertex $\pmv v$ such that  $v \in \mconeof{\pmv v}$ and for all $m\neq \pmv v $ such that $v\in \mconeof m$, then $\pmv v \in \mconeof m$.
We write $v=\pmv v$ if there is no $m$ such that $m\medge v$.
To have an intuition consider a formula $F$, the \ma $\gG=\arof F$ and the formula tree $\ftree F$.
If $\pmv v\neq v$, then the vertex $\pmv v$ corresponds to the root of the smaller subtree of $F$ with root labeled by a modality which contains the node corresponding to $v$. If $\pmv v = v$, then such a node does not exist.
By means of example, in Example~\ref{ex:cones} we have $\pmv a=a$, $\pmv\ldia=\pmv b=\pmv c=\lbox$,  $\pmv d=\ldia$ and $\pmv e=\pmv f=\ldia'$, $\pmv\lbox=\lbox$ and $\pmv \ldia'=\ldia'$.

\begin{theorem}\label{thm:arToFor}
	Let $\gG$ be a  $\labelset$-labeled \ddag.
	If $\gG$ is a \ma, then  there is a formula $F$ such that $\gG=\arof F$.
\end{theorem}
\begin{proof}
	
	We proceed by induction on the size of $\gG$.
	If $\sizeof{\vertices[\gG]}=1$ then if $\lab v \in \atoms$, then $F=a\in \atoms$, if $\lab v=\ldia $ then $F=\ldia \unit$.
	Otherwise, since $\tuple{\vertices[\gG], \iedge[\gG]}$ is a arena, we conclude by Lemma~\ref{lemma:cones} (see \cite{ICP}) that
	\begin{enumerate}
		\item  \label{arena:imp} either 
		every vertex in $\vertices[\gG]\setminus\irof[\gG] $ has a $\iedge$-paths to all roots in $\irof[\gG]$,
		
		\item \label{arena:sum} or 
		$\irof[\gG]$ admits a partition $\irof[\gG]=R_1\uplus R_2$ such that any vertex in $\gG$ has $\iedge$-paths only to roots in one of the two sets.
	\end{enumerate}

	If  \ref{arena:imp} holds, then we define $\gG_2$  as the \ma obtained from $\gG$ taking the vertices in
	$$V_2=\irof[\gG]\cup \big(\bigcup_{v\in \irof[\gG]} \mconeof{v} \big)$$  
	and $\gG_1$ as the \ma over the remaining vertices $V_1=\vertices[\gG]\setminus V_2$.
	Since each vertex in $\gG$ has a path to all the roots in $\irof[\gG]$, 
	then there is a $\iedge$ from any root of $\gG_1$
	to any root of $\gG_2$.
	Since  by definition  $\irof[\gG_2]=\irof[\gG]$, then we have that $\gG=\gG_1\gimp \gG_2$.
	
	If \ref{arena:sum} holds and $\irof[\gG]=R_1\uplus R_2$ with $R_1$ and $R_2$ non-empty sets. 
	Since $\medge$ is modal, we have the following possibilities:
	\begin{enumerate}[(a)]
		\item \label{arIsForm:1}
		if $R_1= \set{v} $ and $v\medge w$ for all $w\in R_2$, then there is no $u$ such that $u\iedge v$. 
		Otherwise $u\iedge v$ and $u\iedge w$ for all $w$ such that $v\medge w$, that is for all $w\in R_2$. This implies that $u\medge w$ for all $w\in \irof[\gG]$, which contradicts the hypothesis \ref{arena:sum}.
		Thus we conclude that $\gG=v\gmod \gG'$ where $\gG'$ is the \ma with vertices $\mconeof{v}$;
		
		\item \label{arIsForm:2}
		if there are no $\medge$-edges between $R_1$ and $R_2$, then $\gG=\gG_1\gsum\gG_2$ where $\gG_1$ and $\gG_2$ are the the \mas with vertices $V_1=\set{v\mid v\iedge^*w \mbox{ for a } w\in R_1}$ and $V_2=\set{v\mid v\iedge^*w \mbox{ for a } w\in R_2}$.
		In fact by definition there are no $\iedge$-edges between vertices in $V_1$ and $V_2$ otherwise by Lemma~\ref{lemma:cones} we should have $R_1=R_2$. Similarly there are no $\medge$-edges between vertices in $V_1$ and $V_2$ since there are no $\medge$-edges between $R_1$ and $R_2$ (by hypothesis) and if there is $v\in V_1\setminus R_1$ and $w\in V_2$ such that $v\medge w$, then by Lemma~\ref{lemma:modalities} $w\notin R_2$ and we should have again $R_1=R_2$;
		
		\item
		otherwise, we pick a $v\in \irof[\gG]\cap \mrof[\gG]$ and define $R_1=\set{v}\cup \set{w\mid v\medge w}$ and $R_2=\irof[\gG]\setminus R_1$.
		If there is no $u\in \irof[\gG]$ such that $v\nmedge u$, then $R_1=\irof[\gG]$ and we conclude by (\ref{arIsForm:1}).
		If $R_2\neq \emptyset$, then we define $V_1=\set{v\mid v\iedge^*w \mbox{ for a } w\in R_1}$ and $V_2=\set{v\mid v\iedge^*w \mbox{ for a } w\in R_2}$ and we conclude by (\ref{arIsForm:2}).	
	\end{enumerate} 
\end{proof}

As result of Lemma~\ref{thm:forToAr} and Theorem~\ref{thm:arToFor}, we have the following correspondence between formulas and \mas:
\begin{theorem}\label{thm:Marena}
	A $\labelset$-labeled \ddag $\gG$ is a \ma iff there is a formula $F$ such that $\gG=\arof F$.
\end{theorem}

We define the \emph{formula isomorphism} as the equivalence relation over formulas $\feq$ generated by the following relations:
\begin{equation}\label{eq:formulaeq}
	\begin{array}{c}
		A\land B \feq B \land A
		\qquad
		A\land (B\land C) \feq (A \land B)\land C
		\\
		(A\land B)\imp  C \feq A\imp (B \imp C)
	\end{array}
\end{equation}

\begin{proposition}\label{prop:formEq}
	If $F$ and $G$ are two formulas and $\feq$ the equivalence relation defined 
	in \Cref{eq:formulaeq} then 
	$$F\feq G  \iff \arof F =\arof G$$
\end{proposition}
\begin{proof}
	By induction using the definition of the \mas operations $\gsum$, $\gimp$ and $\gmod$.
\end{proof}

\def\inv{{\mathsf{in}}}
\def\outv{{\mathsf{out}}}

\def\intv#1{#1^\inv}
\def\extv#1{#1^\outv}

\def\extb#1{#1^{\outv \oddsym}}
\def\extw#1{#1^{\outv \evensym}}
\def\intb#1{#1^{\inv \oddsym}}
\def\intw#1{#1^{\inv \evensym}}

\section{Modal Arena Nets}\label{sec:arenaNets}

In this section we show the correspondence between (linear) proofs in $\IMLLK$ and $\IMLLD$, and respectively $\CK$- and $\CD$-arena nets, that are, modal arenas equipped with a an equivalence relation over vertices satisfying specific topological conditions.

\begin{definition}
	
	A \emph{partitioned modal arena} $\cgG=\tuple {\vertices[\cgG], \iedge[\cgG],\medge[\cgG], \axeq[\cgG]}$ is 
	given by a \ma $\tuple {\vertices[\cgG], \iedge[\cgG],\medge[\cgG]}$ together with an equivalence relation $\axeq[\cgG]$ over vertices  such that:
	
	\begin{itemize}
		\item if $v\in \avertices[\cgG]$ and  $v\axeq[\cgG]w$, then $w\in \avertices[\cgG]$ and $\lab v=\lab w$;
		\item if $v\in \avertices[\cgG]$, then $v\axeq[\cgG]w$ for a unique $w\in\avertices[\cgG]$.
	\end{itemize}
	
	In a partitioned modal arena we represent the equivalence relation $\axeq$ by drawing a (dashed non-oriented blue) edge $\vv1\quad \vw1\axedges{v1/w1}$ between two distinct vertices in the same $\axeq$-class.
	For better readability, we only represent a minimal subset of these edges relying on the fact that $\axeq$ is an equivalence relation. By means of example, if $\set{u,v,w}$ is an $\axeq$-class, we may represent $\vu1\quad\vv1\quad\vw1\axedges{v1/w1,u1/v1}$ omitting the edge between $u$ and $w$.
	
	We say that a formula (or \pformula) $F$ is \emph{associated} to $\cgG$ if $\arof F=\tuple {\vertices[\cgG], \iedge[\cgG],\medge[\cgG]}$.
\end{definition}

\begin{remark}
	If $v$ and $w$ are vertices in a partitioned modal arena $\cgG$ such that $v\axeq[\cgG]w$, 
	then $v\in \mvertices[\cgG]$  iff $w\in\mvertices[\cgG]$.
\end{remark}

If $\gG$ is an arena, we define $\dep v$ as the length of the longest $\iedge$-paths from $v$ to a root $w\in\irof[\gG]$.
The \emph{parity of a vertex}  $v$ is the parity of $\dep v$.
We denote by $\iseven v$ and $\isodd v$ if the parity of $v$ is respectively even or odd.

\begin{remark}\label{rem:toRoot}
	In a arena $\gG$, if $v$ and $w$ are vertices  such that $\dep v=n>0$ and $v\iedge[\gG]^*  w$ and  $w\in \irof[\gG]$, then $v\iedge[\gG]^n w$.
\end{remark}

The \emph{parity} of a $\iedge$-edge $v\iedge w$ is the parity of $\dep w$.
We say that  an edge $v\iedge w$ is a \emph{chord} if there is a vertex $u$ such that 
either $v\iedge u$ and $u\medge w$;
or $u\iedge w$ and $u\medge v$.
By means of example, in the following \mas the edges $a\iedge b$ are chords.
$$
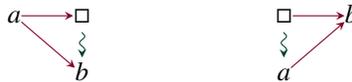
\begin{array}{c@{\qquad}c}
	\va1 & \vlbox1
	\\[1em]
	& \vb1
\end{array}
\modedges{box1/b1}
\dedges{a1/b1, a1/box1}
\qquad\qquad\qquad
\begin{array}{c@{\qquad}c}
	\vlbox1 & \vb1
	\\[1em]
	\va1 & 
\end{array}
\modedges{box1/a1}
\dedges{box1/b1,a1/b1}
$$

We denote by $\oedge[\cgG]$ the set of odd edges which are not chords.

Moreover, we define the set of edges
$$
\omedge[\gG]=\set{vw \mid
	\mbox{ either }\iseven v 
	\mbox{ and } w=\pmv v
	\mbox{ or }\isodd w 
	\mbox{ and } w= \pmv v
}
$$
Note that $\isodd v\omedge \isodd w$ implies $\isodd v\medge \isodd w$, while $\iseven v\omedge \iseven w$ implies $\iseven w\medge \iseven v$.
That is, if $\lbox A$ is right-hand side formula of a sequent (i.e., $\iseven{\lbox A}$), then we have a $\omedge$ from the vertex $\ssbox$ to all the $\iedge$-roots of $\arof{A}$; 
while if $\lbox A$ is left-hand side formula of a sequent (i.e., $\isodd{\lbox A}$), then we have a $\omedge$ from the $\iedge$-roots of $\arof{A}$ to the vertex $\ssbox$.

\begin{definition}\label{def:KarenaNet}
	A partitioned modal arena $\cgG$ is \emph{linked} if every $\axeq[\cgG]$-class is of the form $\set {\isodd v_1, \dots, \isodd v_n,\iseven w}$. 
	This induces the set directed edges $\axlink[\gG]= \set{ (v,w) \mid \isodd v \axeq[\gG] \iseven w}$.
	The \emph{linking graph} $\linkgraph \gG$ of a modal arena is the direct graph with vertices $\vertices[\gG]$ and edges 
	$ \oedge[\gG]\cup \omedge[\gG]\cup\axlink[\gG]$.	
	
	We say that path in $\linkgraph{\gG}$ is \emph{checked} if it starts from a vertex in  $\irof[\gG]\cap \mrof[\gG]$ and it contains no $\axlink$ with source $v$ with $\mconeof{v}\neq \emptyset$.
	
	A \emph{$\CK$-arena net} is a non-empty linked modal arena which satisfies conditions \ref{cond:acyclic}-\ref{cond:K} below:
	\begin{enumerate}
		\item \label{cond:acyclic} $\linkgraph \gG$ is acyclic: every checked path is acyclic;
		
		\item \label{cond:functional} $\linkgraph \gG$ is \emph{functional}: 
		every checked path in $\linkgraph \gG$ from a vertex $\isodd v$ to a root includes a vertex $\iseven w$ such that $v\iedge w$;

		\item \label{cond:functorial} $\cgG$ is \emph{functorial}: if $v\medge w$ and  $w\axeq w'$ 
		then there is $v'$ such that  $v\axeq v'$ 
		and $v'\medge w'$;
		
		\item \label{cond:mod} $\cgG$ has almost all \emph{non-empty modalities}
		\footnote{The only empty modality admitted is a $\isodd{\ldia}$, that is, the $\ssdia$ which corresponds to a $\isodd{\ldia \unit}$ introduced by a $\Kurule$-rule.}:
		if $v\in \mvertices[\cgG]$ and there is no $w\in \vertices[\cgG]$ such that $v\medge w$, then $v\in\dvertices[\cgG]$;

		\item \label{cond:K}  $\cgG$ is \emph{$\CK$-correct}:
		if $\set{\isodd v_1,\isodd v_2, \dots, \isodd v_n, \iseven w} \subset \bvertices[\cgG]\cup \dvertices[\cgG]$ is a $\axeq$-class, then 
		either 
		$v_1,v_2, \dots, v_n,w\in \bvertices[\cgG]$ 
		or
		there is a unique $i$ such that $v_i,w\in  \dvertices[\cgG]$.

	\end{enumerate} 
	A modal arena is a \emph{$\CD$-arena net} if it satisfies Conditions \ref{cond:acyclic}-\ref{cond:functorial} plus the following:
	\begin{enumerate}
		\setcounter{enumi}{5}
		\item \label{cond:modD} $\cgG$ has all \emph{non-empty modalities}:
		if $v\in \mvertices[\cgG]$, then there is $w\in \vertices[\cgG]$ such that $v\medge w$;
		
		\item \label{cond:D}  $\cgG$ is \emph{$\CD$-correct}: 
		if $\set{\isodd v_1,\isodd v_2, \dots, \isodd v_n, \iseven w} \subset \bvertices[\cgG]\cup \dvertices[\cgG]$ is a $\axeq$-class, then 
		either 
		$v_1,v_2, \dots, v_n,w\in \bvertices[\cgG]$ 
		or
		$w\in \dvertices[\cgG]$ there is at most  one $i\in \intset1n$ such that $v_i\in  \dvertices[\cgG]$.
	\end{enumerate}
	
	A \emph{modal arena net} is either a $\CK$- or a $\CD$-arena net.
	An \emph{arena net} is a modal arena net with $\mvertices=\emptyset$. 
	Note that in this case 
	Conditions \ref{cond:functorial}, \ref{cond:mod}, \ref{cond:K}, \ref{cond:modD} and \ref{cond:D}
	are vacuous.
\end{definition}

The intuition for Conditions \ref{cond:K} and \ref{cond:D} is that 
$\axeq$-classes represent either atoms paired by an $\AXrule$, or the set of modalities introduced by a same $\Kbrule$, $\Kdrule$, $\Kurule$ or $\Drule$-rule.
Following this intuition, if 
$\mathsf c = \set{v_0, v_1, \dots, v_n}\subset \mvertices[\cgG]$ is a $\axeq$-class, then the modal arena with vertices 
$\bigcup_{v\in \mathsf c} \mconeof{v}$
corresponds the sub-proof of the premise of any such rule.

\begin{figure*}[t!]
	\newcommand{\varman}[2]{\tuple{{#1}\mid\axeq[#2]}}
	\newcommand{\vvarman}[3]{\tuple{{ #1}\mid\axeq[#2]\cup\axeq[#3]}}
	\def\myskip{\hskip1.9em}
	\begin{center}
		\footnotesize
		\begin{tabular}{c@{\myskip}c@{\myskip}c@{\myskip}c@{\myskip}c@{\myskip}c}
			$\vlinf{}{\axrule}{ \va 1 \myskip \va 2}{} \dedges{a1/a2}\axlbents{a1/a2}$
			&
			$\vlinf{}{\rimprule}{ \cgF  \vdash \cgG \gimp \cgH}{ \cgF , \cgG \vdash \cgH}$
			&
			$\vliinf{}{\limprule}{ \cgF,  \cgJ , \cgG \gimp \cgK \vdash \cgH}{ \cgF \vdash \cgG}{ \cgJ, \cgK \vdash \cgH}$
			&
			$\vliinf{}{\rlandrule}{ \cgF,  \cgI \vdash \cgG \gsum \cgJ}{ \cgF \vdash \cgG}{ \cgI \vdash \cgJ}$
			& 
			$\vlinf{}{\llandrule}{ \cgF, \cgJ \gsum \cgI \vdash \cgK}{ \cgF, \cgJ, \cgI \vdash \cgK}$
		\end{tabular}
		\begin{tabular}{ccc}
			$\vlinf{}{\Kbrule}{ \vvarman{ \ssbox \gmod \cgG_1, \dots,  \ssbox \gmod \cgG_n \vdash \ssbox \gmod \cgH }{\gG}{\nu} }{\varman{ \cgG_1, \dots, \cgG_n \vdash \cgH}{\gG}}$ 
			&
			$\vlinf{}{\Kdrule}{ \vvarman{ \ssbox \gmod \cgG_1, \dots,  \ssdia \gmod \cgG_{i},\dots, \ssbox \gmod \cgG_n \vdash \ssdia \gmod \cgH }{\gG}{\nu} }{\varman{ \cgG_1, \dots, \cgG_n \vdash \cgH}{\gG}}$ 
			\\
			$\vlinf{}{\Kurule}{ \vvarman{ \ssdia, \ssbox \gmod \cgG_1, \dots, \ssbox \gmod \cgG_n \vdash \ssdia \gmod \cgH }{\gG}{\nu} }{\varman{ \cgG_1, \dots, \cgG_n \vdash \cgH}{\gG}}$ 
			&
			$\vlinf{}{\Drule}{ \vvarman{ \ssbox \gmod \cgG_1, \dots,  \ssbox \gmod \cgG_n \vdash \ssdia \gmod \cgH }{\gG}{\nu} }{\varman{ \cgG_1, \dots, \cgG_n \vdash \cgH}{\gG}}$ 
		\end{tabular}
		
		\vspace{10pt}
		where $\axeq[\nu]=\set{\set{x,y}\mid x\mbox{ and }y \mbox{ vertices in the rule conclusion not occurring in the premise} }$
		\vspace{-10pt}
	\end{center}
	\caption{Translation of the sequent rules in $\IMLLK$ and $\IMLLD$ into modal arena nets rules}
	\label{fig:marules}
\end{figure*}

\begin{lemma}\label{thm:arenaCompleteness}
	Let $\X\in \set{\CK,\CD}$.
	If $\provevia{\IMLLX} F$, then 
	there is a $\X$-arena net 
	$\cgG=\tuple{\vertices[\gG],\iedge[\gG], \medge[\gG],\axeq[\gG]}$ such that 
	$\arof F= \tuple{\vertices[\gG], \iedge[\gG], \medge[\gG]}$.
\end{lemma}
\begin{proof}
	Let $\pi$ be a derivation of $F$ in $\IMLLX$.
	We proceed by induction translating the derivation $\pi$ of the formula $F$ in a derivation of a modal arena $\cgG$ (with associated formula $F$) in the system described by the rules in  \Cref{fig:marules}.
	
	By definition, each rule in $\IMLLX$ preserves $\X$-arena net conditions, that is, if the premises of a rule are $\X$-arena nets, then the conclusion is. 
	In particular, Condition \ref{cond:K} fails for the rule $\Drule$.
	Similarly, each rule except $\Kurule$ preserves the $\CD$-arena net conditions.
\end{proof}

\begin{figure*}[!t]
	\centering
	$
	\begin{array}{c@{\quad}c@{\quad }c}
		\begin{array}{cccccc}
			\vc1 &				&				& \vb0	
			\\[1em]
			\vc0 & \vlbox1_1 & \vlbox0_0 & \vb1
			\\[1em]
			\qquad		&	\va1	&	\va0	&\qquad
			\axlbents{box1/box0,a1/a0, b1/b0}
			\modedges{box1/a1,box0/a0}
			\dedges{c0/box1,c0/a1,a0/b1,box0/b1,b1/b0,c1/b0}
			\multidedges{box1,a1}{a0,box0}
			\axedges{c1/c0}
		\end{array}
		%
		& \overset{\partial}{\rightsquigarrow} &
		%
		\begin{array}{cccccc}
			\vc1 &					&&		&							& \vb0	
			\\[1em]
			\vc0 & \vlbox1_1^\outv 	 & 	\vlbox3_1^\inv & \vlbox2_0^\inv&\vlbox0_0^\outv & \vb1
			\\[1em]
			&&	\va1	&	\va0
		\end{array}
		\axlbents{a1/a0, b1/b0}
		\dedges{c0/box1,box0/b1,b1/b0,c1/b0}
		\dedges{box3/a1,a0/box2, a1/a0}
		\axedges{box1/box3,box2/box0,c1/c0}
		\dlbent{box1}{box0}
		%
		\\[3em]
		\Updownarrow & & \Updownarrow
		\\[-.5em]
		%
		{\scriptsize
			\vlderivation{
				\vliin{}{\limprule}{ c ,  {((c\imp \lbox a) \imp \lbox a )\imp b } \vdash  b}
				{\vlin{}{\rimprule}{ c \vdash {(c\imp \lbox a) \imp \lbox a}}
					{\vliin{}{\limprule}{ c,  {c\imp \lbox a} \vdash { \lbox a}}
						{\vlin{}{\AXrule}{ c \vdash c}{\vlhy{}}}
						{\vlin{}{\Kbrule}{\lbox  a \vdash  \lbox  a}{
								\vlin{}{\AXrule}{ a \vdash  a}{\vlhy{}}}}}}
				{\vlin{}{\AXrule}{ b \vdash  b}{\vlhy{}}}
			}
		}
		%
		& \overset{\partial}{\rightsquigarrow} &
		%
		{\scriptsize
			\vlderivation{
				\vliin{}{\limprule}{ c,  {(\intv\lbox_1 \imp a) \land (a\imp \intv\lbox_0)} ,{((c\imp \extv \lbox_1 )\imp \extv \lbox_0 )\imp b} \vdash  b}
				{\vlin{}{\rimprule}{c, {(\intv \lbox_1 \imp a) \land (a\imp \intv\lbox_0 ) } \vdash {{(c\imp \lbox_1^\outv )\imp\lbox_0^\outv}}}
					{\vliin{}{\limprule}{c, {{c\imp \extv \lbox_1}} , {(\intv \lbox_1 \imp a) \land (a\imp \intv \lbox_0 )} \vdash \extv \lbox_0}
						{\vlin{}{\AXrule}{ c \vdash  c}{\vlhy{}}}
						{\vlin{}{\llandrule}{\extv \lbox_1 , {(\lbox_1^\inv \imp a) \land (a\imp \lbox_0^\inv )}\vdash \extv \lbox_0}	
							{{\vliin{}{\limprule}{\extv \lbox_1 , {\intv\lbox_1 \limp a} , {{a \limp \intv\lbox_0}} \vdash \extv\lbox_0}
									{\vliin{}{\limprule}{\extv\lbox_1 , {\intv\lbox_1\imp a}\vdash   a}
										{\vlin{}{\AXrule}{\intv\lbox_1\vdash \extv \lbox_1}{\vlhy{}}}
										{\vlin{}{\AXrule}{ a \vdash  a}{\vlhy{}}}}
									{\vlin{}{\AXrule}{\intv\lbox_0 \vdash \extv\lbox_0}{\vlhy{}}}}}}
				}}
				{\vlin{}{\axrule}{b \vdash b}{\vlhy{}}}
			}
		}
	\end{array}
	$
	\caption{A $\K$-arena net $\cgG$ with associated formula $ (c \land {((c\imp \lbox a) \imp \lbox a )\imp b }) \imp b$, its corresponding arena $\der\cgG$, the derivations associated to  $\cgG$ and $\der\cgG$
	}
	\label{fig:NetToDer1}
\end{figure*}

\begin{lemma}\label{thm:arenaSoundness}
	Let $\X\in\set{\CK,\CD}$ and $\cgG$ a modal arena with associated formula $F$. 
	If $\cgG$ is a $\X$-arena net, then  $\provevia{\IMLLX}F$.
\end{lemma}
\begin{proof}
	We prove the theorem for $\CK$-arena nets.

	To prove this theorem we define from the $\CK$-arena net $\cgG$, with associated formula $F$, an arena net $\der\cgG$ with associated formula $\der F$.
	We then use use of the result in \cite{ICP} on (non-modal) arena nets to produce an $\IMLL$-derivation of $\der F$.
	Then we conclude by showing how to define a $\IMLLX$-derivation of $F$ using the $\IMLL$-derivation of $\der F$.

	\textit{Step 1: definition of $\der\cgG$.}
	If $\cgG=\tuple{\vertices[\gG], \iedge[\gG], \medge[\gG], \axeq[\gG]}$ is a $\CK$-arena net, 
	we write $\sameframe vw$ either if $\pmv {v} \axeq \pmv{w}$,  or if $v=\pmv v $ and $w=\pmv w$,  that is, $\sameframe vw$ iff both $v$ and $w$ belongs to the scope of a same modality or in the scope of no modality.
	
	We define the arena $\der\cgG$ 
	by removing all $\medge$-edges in $\cgG$ and keeping only the $\iedge$ between vertices $v,w\in \vertices[\cgG]$ such that $\sameframe vw$.
	Then we replace each modal vertex $v$ by a pair of $\axlink$-linked vertices $\intv v,\extv v$ in such a way that the vertex $v^\inv$ keeps track of the subformulas of the modality, while $\extv v$ is a placeholder to keep track of the interaction of the subformulas with the context.
	
	Formally we define $\der\cgG=\tuple{\der{\vertices[\gG]}, \der{\iedge[\gG]\cup \medge[\gG]},{\axeq[\der{\gG}]}}$ 
	by:
	\begin{itemize}
		\item 
		$\der{\vertices[\gG]}= \avertices[\gG] \cup \set{\intv v, \extv v \mid v\in  \mvertices[\gG]}$

		\item 
		$\der{\iedge[\gG]\cup \medge[\gG]}$  is the union of the following sets where we assume 
		$u,v\in \avertices[\cgG]$ and $\isodd l,\iseven r,m,n,p\in \mvertices[\cgG]$:
		
		$$
		\begin{array}{l}
			\set{(u,v)\mid \sameframe{u}{v} \mbox{ and } u \iedge v}
			\\
			\set{(\extv l, \extv r)\mid l\axlink r}
			\\
			\set{(u,\intv r),(\intv l,v) \mid 
				l=\pmv v, r=\pmv u}
			\\
			\set{(u,\extv m),(\extv m, v) \mid 
				u\iedge m \iedge v,
				m\sameframesymbol u\sameframesymbol v}
			\\
			\set{(\extv m,\extv n),(\extv n, \extv p) \mid 
				m\iedge n\iedge p, 
				m\sameframesymbol n\sameframesymbol p}
		\end{array}
		$$

		\item 
		${\axeq[\der{\cgG}]}$ is defined as:
		$$\begin{array}{cl}
			v\axeq[\der{\cgG}]w 
			&
			\mbox{ if } v,w\in \avertices[\cgG]\subset \der{\vertices[\cgG]} \mbox{ and } v\axeq[\cgG]w
			\\
			\intv v \axeq[\der{\cgG}] \extv v 
			&
			\mbox{ for each } v\in \mvertices[\cgG]
		\end{array}$$

	\end{itemize}
	See the first line of \Cref{fig:NetToDer1,fig:NetToDer2} for running examples.
	
	We observe that if $\set{\iseven v_0, \isodd v_1, \dots, \isodd v_n}$ is a $\axeq[\cgG]$-class of modal vertices, then a \pformula associated to $\cgG$ is of the form 
	$$H=H\cons{\iseven{\lab {v_0}A_0}}\cons{\isodd{\lab{v_1}A_1}} \cdots \cons{\isodd{\lab{v_n}A_n}}$$
	for an $(n+1)$-ary context $H\cons{}\cdots\cons{}$.
	In this case, a \pformula associated to the arena $\der \cgG$ is of the form
	{
		$$
		\der H
		=
		\der H
		\cons{\iseven{{\extv v_0}}}
		\cons{\isodd{{\extv{v_1}}}}
		\cdots
		\cons{\isodd{{\extv{v_n}}}}
		\cons{\isodd{H_{\mathsf c}}}
		$$
	}
	with $\der H \cons{ }\cdots \cons{ }$ is an $(n+2)$-ary context, $\intv v_i,\extv v_i$ are fresh propositional variables for all $i\in\set{0,\dots, n}$ and  
	{\small$$
		\isodd H_{\mathsf c}=\isodd{\Big(\Big(\big(v_1^\inv \limp \der {\isodd A_1} \ltens  \cdots \ltens  v_n^\inv \limp \der {\isodd A_n}\big)\imp  \der {\iseven A_0}\Big)\limp v_0^\inv\Big)}
		$$}

	\def\path{\mathsf{p}}
	\def\derpath{\der\path}
	\textit{Step 2: prove that $\der\cgG$ is an arena net.}

	We observe that, by definition of $\der\cgG$, every path $\derpath$ in $\linkgraph{\der\cgG}$ can be constructed from a checked path $\path$ in $\linkgraph\cgG$ by induction: 
	\begin{itemize}
		\item the empty path is a path in both  $\linkgraph\cgG$ and $\linkgraph{\der\cgG}$;
		
		\item if $\path=v \cdot \path'$ then 	
		\begin{itemize}
			\item if ${v}\in \avertices[\cgG]$, then $\derpath=v \cdot \derpath'$;
			\item if $\isodd{v}\in\mvertices[\cgG]$, then $\derpath=\extv v \cdot v^\inv \cdot \derpath'$;
			\item if $\iseven v\in \mvertices[\cgG]$, then $\derpath=v^\inv  \cdot \extv v  \cdot \derpath'$;
		\end{itemize}
	\end{itemize}
	We remark that the parity of atomic vertices is preserved by $\partial$, while  the parity of a modal vertex $v\in\vertices[\cgG]$ is the same of  the corresponding vertex $\extv v \in \vertices[\der{\cgG}]$.
	Since if $\isodd v$ then $\extv v \axlink v^\inv $, and  if $\iseven v$ then $v^\inv \axlink \extv v $, then 
	we have that in $\der{\cgG}$ an even (odd) vertex may occur  only in a even (odd) position in a path in $\linkgraph{\cgG}$.
	We conclude since from any path in $\linkgraph{\der\cgG}$ we obtain a path in $\linkgraph\cgG$ by replacing every subpath $\extv v \axlink v^\inv $ and $v^\inv \axlink \extv v $ by a the corresponding modal vertex $v$ in $\cgG$.
	
	By this correspondence between checked paths in $\linkgraph{\cgG}$ and paths in $\linkgraph{\der\cgG}$ we conclude that $\linkgraph{\der\cgG}$ is acyclic and functional. That is, $\der\cgG$ is an arena net.

	\textit{Step 3: construct the derivation associated to $\der\cgG$.}
	Since $\der\cgG$ is an arena net, then we apply the result in \cite{ICP} to produce a derivation in $\IMLL$  of the formula $\der F$.
	In such a derivation, by functionality and functoriality of $\cgG$, 
	whenever $v$ and $w$ are modal vertices such that $v\axlink[\cgG]w$, then
	if a path in $\linkgraph{\der \cgG}$ contains $v^\inv$, then it also contains $\extv v$, $w^\inv$, $\extv w$.
	This means that  if $\mathsf c= \set{\iseven v_0,\isodd v_1, \dots, \isodd v_n}$ is an $\axeq[\cgG]$-class of vertices in $\medge[\cgG]$, then  any derivation of $\der F$ in $\IMLL$ contains a subderivation of the sequent 
	$
	\extv v_1, \dots, \extv v_n, \spalletto{\isodd H_c} \vdash \extv v_0
	$
	of the following form
	{\scriptsize
		$$
		\vlderivation{
			{\vliin{}{\llimpr}
				{\extv v_1 , \dots, \extv v_n, \big( {\bigwedge_{i=1}^n (\intv v_i\imp \der {A_i}) \imp \der{A_0}} \big)\imp \intv v_0 \vdash \extv v_0 }
				{\vlin{}{\rimprule}{\extv v_1 , \dots, \extv v_n \vdash {\bigwedge_{i=1}^n (\intv v_i\imp \der {A_i}) \imp \der{A_0}}}{
						\vliq{}{\llandrule}{\extv v_1 , \dots, \extv v_n, {\bigwedge_{i=1}^n (\intv v_i \imp \der {A_i})} \vdash {\der{A_0}}}{ 
							\vliq{}{\limprule}{ \extv v_1 , \dots, \extv v_n, {\intv v_1 \imp \der {A_1}}, \dots, {\intv v_n \imp \der{A_n}} \vdash {\der{A_0}}}
							{
								\vlhy{
									\vlderivation{\vlin{}{\AXrule}{\extv v_1\vdash  \extv v_1}{\vlhy{}}}
									\; \cdots \; 
									\vlderivation{\vlin{}{\AXrule}{\extv v_n \vdash \extv v_n}{\vlhy{}}}
									\qquad
									\vlderivation{\vlpr{}{\IMLL}{{\der {A_1}}, \dots, {\der{A_n}}, \vdash {\der{A_0}}}}
					}}}}
				}
				{\vlin{}{\axrule}{\intv v_0 \vdash \extv v_0}{\vlhy{}}}}}
		$$
	}
	
	In order to construct a derivation in $\IMLLX$ of the formula $F$  it suffices to proceed by induction over the number of 
	$\axeq[\cgG]$-classes of modal vertices.
	Starting from the top of the derivation, we replace every such subderivation in the derivation of $\der{F}$ with an application of a $\Kbrule$-, $\Kdrule$- or $\Kurule$-rule, we remove all the occurrences of the formula  $\spalletto{H_{\mathsf c}}=\big( {\bigwedge_{i=1}^n (\intv v_i\imp \der {A_i}) \imp \der{A_0}} \big)\imp \intv v_0$ in the derivation, and we replace for each $i\in \intset 0n$ the atom $\intv v_i$ with the corresponding formula $\lab {v_i}A_i$  
	as shown in \Cref{fig:linearizedK}. For some example, refer to the lower line of \Cref{fig:NetToDer1,fig:NetToDer2}.

	The proof for $\CD$-arena nets is similar by considering the rule $\Drule$ instead of $\Kurule$.
\end{proof}


\begin{figure*}[h]
	\centering
	$
	\begin{array}{c@{\quad}c@{\quad }c}
		%
		%
		\begin{array}{cc@{\qquad}cccccccc}
			&{\vlbox1}_1  &  					&			&\\
			&					&{\vlbox2}_2 &			&{\vlbox0}_0
			\\[1em]
			{\va2} 	&					& {\va1} \\
			& {\vb1} 		&					&			&{\vb2}
		\end{array}
		\dedges{a2/b1}
		\modedges{box1/b1,box2/a1,box0/b2}
		\multidedges{b1,box1,a1,box2}{b2,box0}
		\axrbents{b1/b2}
		\axedges{a1/a2}
		\axlbents{box1/box0,box2/box0}
		\axedges{box1/box2}
		%
		%
		& \overset{\partial}{\rightsquigarrow} &
		%
		%
		\begin{array}{cc@{\qquad}cccccccc}
			{\vlbox4}_1^\outv&						&									&  						&			&					&{\vlbox3}_0^\outv	\\
			& {\vlbox1}_1^\inv &{\vlbox5}_2^\outv		&{\vlbox2}_2^\inv &			&{\vlbox0}_0^\inv \\
			\\
			{\va2} 		&					&							& {\va1} \\
			& {\vb1} 		&							&						&			&{\vb2}
		\end{array}	
		\dedges{a2/b1, b1/b2,a1/b2}
		\dedges{box1/b1,box2/a1,b2/box0}
		\dedges{box0/box3,box5/box3, box4/box3}
		\axedges{a1/a2,box5/box2, box4/box1}
		\axrbents{b1/b2,box0/box3}
		%
		%
		\\[4em]
		\Updownarrow & & \Updownarrow
		\\[-.5em]
		%
		%
		{\scriptsize
			\vlderivation{\vlin{}{\Kbrule}{{\lbox(a\limp b)}, \lbox  a \vdash  \lbox  b}
				{\vliin{}{\limprule}{{a\imp b},  a \vdash a}{\vlin{}{\AXrule}{ a, \vdash a}{\vlhy{}}}{\vlin{}{\AXrule}{ b, \vdash b}{\vlhy{}}}}
		}}
		& \overset{\partial}{\rightsquigarrow} &
		%
		%
		{\scriptsize
			\vlderivation{
				{\vliin{}{\limprule}
					{\extv \lbox_1 , \extv\lbox_2, {\big(\big( (\intv \lbox_1 \limp (a\imp b)) \land ( \intv \lbox_2\imp a) \big) \imp b\big)\imp {\intv\lbox_0} }\vdash \extv \lbox_0 }
					{
						\vlin{}{\rimprule}{\extv\lbox_1 ,\extv \lbox_2 \vdash { \big( (\lbox_1^\inv\imp (a\limp b)) \land ( \lbox_2^\inv\imp a) \big) \imp b}}{
							\vlin{}{\llandrule}{ \extv\lbox_1 , \extv\lbox_2, {(\lbox_1^\inv\limp (a\imp b)) \land ( \lbox_2^\inv\imp a)} \vdash  b}{
								\vliq{}{\limprule}{ \extv\lbox_1, \extv\lbox_2, {\intv\lbox_1\imp (a\limp b)}, {\intv\lbox_2\imp a} \vdash  b}
								{
									\vlhy{
										\vlderivation{\vlin{}{\AXrule}{\extv\lbox_1\vdash \intv\lbox_1}{\vlhy{}}}
										\quad
										\vlderivation{\vlin{}{\AXrule}{\extv\lbox_2\vdash  \intv\lbox_2}{\vlhy{}}}
										\qquad
										\vlderivation{\vliin{}{\limprule}{{a\imp b},  a \vdash a}{\vlin{}{\AXrule}{ a\vdash a}{\vlhy{}}}{\vlin{}{\AXrule}{ b \vdash b}{\vlhy{}}}}
						}}}}
					}
					{\vlin{}{\AXrule}{\intv \lbox_0 \vdash  \extv \lbox_0}{\vlhy{}}}}}
		}
	\end{array}
	$
	\caption{A $\K$-arena net $\cgG$ with associated formula $\lbox(a\imp b)\imp(\lbox a\imp\lbox b)$ (the axiom $\krule_1$), its corresponding arena $\der\cgG$, the derivations associated to  $\cgG$ and  $\der\cgG$
	}
	\label{fig:NetToDer2}
\end{figure*}

\begin{figure*}[ht]
	\centering
	$
	\vlderivation{
		\vlde{\dD			
			\cons{\extv\lbox_0}
			\cons{\extv\lbox_1}
			\cdots
			\cons{\extv\lbox_n}
			\cons{\spalletto {H_{\mathsf c}} }
		}{\IMLL}
		{{\der F
				\cons{\extv\lbox_0}
				\cons{\extv\lbox_1}
				\cdots
				\cons{\extv\lbox_n}
				\cons{\spalletto{ H_{\mathsf c} }}}
		}
		{
			\vlde{\der{\dD'}}{\IMLL}
			{\extv\lbox_1 , \dots, \extv\lbox_n , \spalletto {H_{\mathsf c}}\vdash  \extv \lbox_0}
			{
				\vlhy{
					\vlderivation{
						\vlpr{{\dD''}}{\IMLL}{ {A_1}, \dots, {A_n}, \vdash {A_0}}}
		}}}
	}
	\rightsquigarrow
	\qquad
	\vlderivation{
		\vlde{\dD
			\cons{\lbox  A_0}
			\cons{\lbox  A_1}
			\cdots
			\cons{\lbox  A_n}
			\cons{\emptyset}		
		}{\IMLLX}{F\cons{\lbox_0 A_0}\cons{\lbox_1 A_1}\cdots\cons{\lbox_n A_n}}
		{
			\vlin{}{\Kbrule}{\lbox  A_1, \dots, \lbox A_n \vdash \lbox  A_0}{
				\vlpr{\dD''
				}{\IMLLX}{A_1, \dots,  A_n \vdash  A_0}
			}
		}
	}
	$
	\caption{An example of the construction of the derivation of $F$ from the derivation of $\der F$ assuming that in $\cgG$ there is only one $\axeq$-class of the form  $\set{\lbox_0, \dots, \lbox_n}$}
	\label{fig:linearizedK}
\end{figure*}

We summarize the results of this section (Lemmas~\ref{thm:arenaCompleteness} and \ref{thm:arenaSoundness}) by the following
\begin{theorem}\label{thm:linSeq}
	Let $\X\in\set{\CK,\CD}$ and $\cgG$ be a modal arena with associated formula $F$, then
	$$\cgG \mbox{ is a  $\X$-arena net} \iff  \provevia{\IMLLX}F$$
\end{theorem}

\section{Skew fibrations}\label{sec:skew}

In this section we define specific maps between \mas to model the application of the deep inference rules in \Cref{fig:deepRules}.

If $v,w$ are two vertices in an \ma, a \emph{meeting point} of $v$ and $w$ is a vertex $u$ such that $ v\iedge^* u$ and $ w \iedge^* u $, and such that if there is $u'$ such that 
$ v\iedge^* u'$ and $ w \iedge^* u' $, then $u\iedge^* u'$.
The \emph{meeting  depth} 
of $v$ and $w$ is the minimum of the depth of their meeting point or $-1$ if no such vertex exists.
Two distinct vertices $v$ and $w$ are \emph{conjunct}, denoted $v\conj w$ if  their meeting depth is odd;
they are \emph{disjunct}, denoted $v\disj w$ if their meeting depth is even.

\begin{definition}[Skew Fibration]
	An \emph{arena homomorphism} is 
	either
	a map $\cempty_\gG \colon \emptyset \to \gG$ from the empty \ddag to an \ma $\gG$, 
	or
	a map $\cf \colon \gH \to \gG$ between two \mas $\gH$ and $\gG$ mapping $\vertices[\gH]$ to $\vertices[\gG]$ in such a way it preserves:
	
	\begin{tabular}{c@{\quad : \quad}l}
		$\iedge$
		&
		if $v\iedge[\gH] w$ then $\cf(v) \iedge[\gG] \cf(w)$;
		\\
		$\medge$
		&
		if $v\medge[\gH] w$ then $\cf(v) \medge[\gG] \cf(w)$;
		\\
		$\depsym$
		&
		$\dep v = \dep{\cf(v)}$;
		\\
		$\labsym $
		&
		$\lab v = \lab {\cf(v)}$.
	\end{tabular}

	An arena homomorphism is \emph{modal} whenever:
	\begin{itemize}
		\item 
			if $\cf(v)\medge[\gG] \cf(u)$, then  $w \medge[\gH] u$ and $\cf(v)=\cf(w)$ for a $w\in\vertices[\gG]$.
	\end{itemize}
	
	An \emph{(even) skew fibration}  is a modal arena homomorphism $\cf \colon \gH \to \gG$ which:
	\begin{itemize}
		\item preserves $\conj$: if $v\conj[\gH] w$ then $\cf(v) \conj[\gG] \cf(w)$;
		
		\item is a \emph{skew lifting}: 
		if $\cf(v)\conj[\gG] w $, then there exists $u$ with $v \conj[\gH] u$ and $\cf(u)\nconj[\gG]w$;
	\end{itemize}
	
	An \emph{odd skew fibration}  is either a map $\cempty_\gG \colon \emptyset \to \gG$, or a modal arena homomorphism $\cf \colon \gH \to \gG$ which:
	\begin{itemize}
		\item preserves $\disj$: if $v\disj[\gH] w$ then $\cf(v) \disj[\gG] \cf(w)$;
		
		\item is a \emph{odd skew lifting}:
		if $\cf(v)\disj[\gG] w $, then there exists $u$ with $v \disj[\gH] u$ and $\cf(u)\ndisj[\gG]w$;
	\end{itemize}
\end{definition}

\begin{remark}
	In \cite{ICP} the definition of skew fibration only demands the weaker \emph{root preserving} condition (that is, if $v\in \irof[\gH]$ then $\cf(v)\in \irof[\gG]$) instead of the depth preserving condition we propose here.
	However, in the same paper it is proven that root preserving is equivalent to the depth preserving for even and odd skew fibrations between arenas.
\end{remark}

In order to prove the results in this section, it is useful to highlight the correlations between mutual position of nodes in the formula tree $\ftree F$ of a formula $F$ and the presence of $\iedge$- or $\medge$-edges between the corresponding vertices in  $\arof F$.

\begin{definition}
	If $F$ is a formula, the \emph{formula tree} of $F$ is the tree $\ftree F$ with nodes labeled by $\imp$, $\land$ or $\unit$ symbols,  atoms, and  modalities occurring in $F$. It is defined inductively as follows:
	
	\begin{itemize}
		\item if $F=a$, then $\ftree F$ is the with a single node labeled by $a$;
		
		\item if $F=A\imp B$ (respectively $F=A\land B$), then $\ftree F$ is the tree with root labeled by $\imp$ (respectively $\land$) with children the roots of $\ftree A$ and $\ftree B$;
		
		\item If $F=\lbox A$ ($F=\ldia A$) for a formula $A$, then $\ftree F$ is the tree with root labeled by $\lbox$ (respectively $\ldia$) which has as one child the root of $\ftree A$;
		
		\item If $F=\ldia \unit$, then $\ftree F$ is the tree with root labeled by $\ldia$ and a single child labeled by $\unit$;
	\end{itemize}
\end{definition}

\begin{example}\label{ex:ftree}
	Let $F= \lbox(\ldia\unit \imp (a\land b))   \imp ( (c\imp d) \land e)$ be a formula. The  formula tree $\ftree F$ is the following
	$$
	\begin{array}{ccccccccccccc}
		&				&		&			&	\vimp3		&								\\
		&\vlbox1 &			&			&  			&		&				&\vand2	\\[.5em]	
		&\vimp1	&			&				&		&		&\vimp2		&& \ve1		\\		
		\vldia1		 &				&			&\vand1&			&\vc1& &\vd1					\\[.5em]	
		\vunit1 	&				&\va1	&			&\vb1	
	\end{array}
	\treeedges{imp3/box1,imp3/and2, box1/imp1,imp1/dia1,dia1/unit1,imp1/and1,and1/a1,and1/b1,and2/imp2, and2/e1,imp2/c1,imp2/d1}
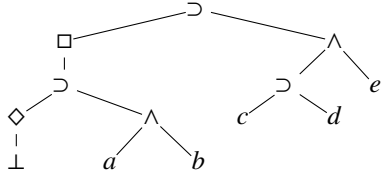
	$$
\end{example}

\begin{definition}
	In a formula tree $\ftree F$, we call a node the  \emph{left-hand side child} (\emph{right-hand side child}) of a $\imp$-node if it corresponds to the root of the left-hand side (respectively the right-hand side) subformula formula of the implication.
	
	If $v$ is a node of a formula tree $\ftree F$, we say that a node $w$ is a \emph{rightmost descendant} of $v$ if there is a path from $v$ to $w$ in $\ftree F$ containing no left-hand side child of any $\imp$-node.
	If $v$ is a $\imp$-node of a formula tree $\ftree F$, we say that a node $w$ is a \emph{second-rightmost descendant} of $v$ if it is a rightmost descendant of its left-hand side child.	
\end{definition}

By means of example, consider the formula tree of $F= (\lbox(\ldia\unit \imp (a\land b)))   \imp ( (c\imp d) \land e)$
given in Example~\ref{ex:ftree}. 
The left-hand side child of the  root of $\ftree F$ is the root of $\ftree{{\lbox(\ldia\unit \imp (a\land b))}}$ while its the right-hand side child is the root of $\ftree{{(c\imp d) \land e}}$.
The set of rightmost and second-rightmost nodes of the root $\imp$ are respectively  $\set{d,e}$ and $\set{\lbox, a,b}$.

\begin{remark}\label{rem:edgesAndTrees}
	Let $F$ be a formula and $\ftree F$ the formula tree of $F$. 
	If we identify the atom or modality $x$ occurring in $F$ with the corresponding the node of  $\ftree F$ and with the  unique $x$-labeled vertex $\singlevertex[x]$ in $\arof F$, then we have the following correspondence:
	\begin{itemize}
		\item $\singlevertex[a]\iedge[\arof F] \singlevertex[b]$  iff the least common ancestor of $a$ and $b$ in $\ftree F$ is a $\imp$, and $a$ and $b$ are respectively a second-rightmost descendant and a rightmost descendant of the least common ancestor;
		
		\item $\singlevertex[m]\medge[\arof F] \singlevertex[x]$ with iff $x$ is a rightmost descendant of $m\in \set{\lbox, \ldia}$.
	\end{itemize}
	Moreover, $\dep x$ is equal to the number of left-hand side children of $\imp$-nodes occurring in  the path from the root of $\ftree F$ to the node $v$.
\end{remark}

\begin{lemma}\label{lemma:skewComp}
	The composition of two skew fibrations is a skew fibration.
\end{lemma}
\begin{proof}
	By definition of skew fibration (see \Cref{app:proofs:skew}).
\end{proof}
\begin{proof}
	By definition, the composition of two modal skew fibrations preserves $\iedge$, $\medge$, $\axeq$ and $\depsym$.
	Then the preservation of $\conj$ and the skew lifting condition of the composition are guaranteed as consequence of the preservation of $\depsym$ and $\iedge$.
	Similarly, the modal condition of the composition is guaranteed as consequence of the preservation of $\medge$.
\end{proof}

\begin{remark}
	Note that the proof of Lemma~\ref{lemma:skewComp} makes crucial use of the fact that
	we talk about arena homomorphisms, and an arena is always associated
	to a formula. In classical logic~\cite{hughes:pws,acc:str:CPK} a
	skew fibration is defined as a homomorphism between arbitrary
	graphs, and the composition of skew fibrations is only a skew
	fibration if those graphs are associated to formulas.
\end{remark}

We are now able to prove the correspondence between $\DOWN{\PIp}$ derivations from a \pformula $H'$ to $H'$  and skew fibrations between their corresponding arenas.

\begin{lemma}\label{thm:derToSkewfib}
	For any \pformulas $H'$ and $H$, if $H' \provevia{\DOWN{\PIp}} H$, then there is a skew fibration $\cf \colon \arof{H'} \to \arof{H}$.
\end{lemma}
\begin{proof}
	
	If we prove that for all $\rho\in\set{\deep\wdrule, \deep\wtrule ,\deep\wirule,\deep\crule}$ if  $\vlinf{}{\rho}{H}{H'}$ then there is a skew fibration $\cf \colon \arof{H'} \to \arof{H}$, we can conclude by Lemma~\ref{lemma:skewComp}.
	We proceed by case analysis.
	
	If $\rho= \deep\wtrule$, then $H'= \Gamma \cons{\isodd A}$ and $H= \Gamma \cons{\isodd{A \ltens B}}$.
	In particular, we can obtain $\ftree H$ from $\ftree{{H'}}$
	by
	removing the formula subtree $\ftree A$
	and replace its root 
	with a $\ltens$-node with children the roots of  $\ftree A$ and $\ftree B$.
	After Remark~\ref{rem:edgesAndTrees}, the arena homomorphism $\cf \colon \arof{H'} \to \arof{H}$ preserves 
	$\iedge$, $\medge$, $\axeq$, $\depsym$ and $\conj$ by definition. 
	Moreover, since the map is injective, modal  condition is trivially satisfied, while the skew lifting immediately follow by the fact that the roots of $\arof{\isodd A}$ have odd depth in $\arof H$.
	Thus $\cf$ is a skew fibration.

	If $\rho= \deep \wirule$, conclude by a similar reasoning. In this case, we obtain $\ftree H$ from $\ftree{{H'}}$
	by
	removing the formula subtree $\ftree A$
	and replace its root 
	with a $\limp$-node with left-hand side child the root of  $\ftree B$  and right-hand side child the root of $\ftree A$.
	Since the root of $\ftree A$ is the right-hand side child of a $\limp$, arena homomorphism $\cf \colon \arof{H'} \to \arof{H}$  is a skew fibration since it is injective and no new $\iedge$ or $\medge$ are added between the vertices in $\arof H$ image of vertices in $\cf$.

	If $\rho= \deep \wdrule$, conclude by a similar reasoning. In particular, it suffices to replace in $\ftree{H'}$ a leaf labeled by a $\unit$ with the tree of the weakened formula.

	If $\rho= \deep \crule$, then $H'= \Gamma\cons{\isodd{(A_1 \ltens A_2)}}$ and $H= \Gamma\cons{\isodd{A}}$, that is, 
	and $\ftree {H'}$ has the  two identical subtrees $\ftree F$ and their roots are children of a same node labeled by   $\ltens$. 
	Thus  $\ftree H$ can be obtained from $\ftree{H'}$ by removing the node labeled by $\ltens$ and replace it with the root of $\ftree F$.
	After Remark~\ref{rem:edgesAndTrees}, the arena homomorphism $\cf \colon \arof{H'} \to \arof{H}$ preserves 
	$\iedge$, $\medge$, $\axeq$ and $\depsym$. Moreover $\cf$ is surjective, then it is a skew lifting and preserves $\conj$.
	Moreover $\cf$ is modal since it is injective on $\Gamma\cons{}$ and whenever $\cf(v)\medge \cf(u)$, then either $v$ and $u$ are both vertices in $\arof{A_1}$ or in $\arof{A_2}$, or we have $v'$ and $u'$ such that $\cf(v)=\cf(v')$, $\cf(u)=\cf(u')$, $v\medge[\arof{A_i}] u'$, $v \medge[\arof{A_j}] v$ with $i,j\in \set{1,2}$, $i\neq j$.
	We conclude that $\cf$ is a skew fibration.
\end{proof}

In order to prove the converse result, we need some additional lemmas.

\begin{lemma}\label{lemma:arHomSum}
	If $\cf\colon \gH \to \gG$ is an arena homomorphism and  
	$\gG=\gG_1\gsum \gG_2$,
	then $\cf=\cf_1\gsum \cf_2$ with 
	$\cf_1\colon \gH_1 \to \gG_1$
	and 
	$\cf_2\colon \gH_2 \to \gG_2$
	arena homomorphisms
	for some $\gH_1, \gH_2$ such that $\gH=\gH_1\gsum \gH_2$.
\end{lemma}
\begin{proof}
	Since $\cf$ preserves $\iedge$, then if $v\iedge^*w$ for a $w\in \irof[\gG]$ then $\cf(v)\iedge^*\cf(w)$.
	Thus if $\gG= \gG_1\gsum \gG_2$,  then there is a partition\footnote{As remarked in the proof of \Cref{thm:Marena}, in construction such partition, because of $\medge$-coherence, whenever $v\medge w$ then  $v$ and $w$ belong to the same subset.} $\irof[\gG]=\irof[\gG_1]\uplus \irof[\gG_2]$. 
	Then we can define  $\vertices[\gH_1]$ and $\vertices[\gH_2]$ as the sets of vertices of $\gH$ which images by $\cf$ admit a $\iedge$-path to a vertex in $\irof[\gG_1]$ and $\irof[\gG_2]$ respectively.
	The \mas $\gH_1$ and $\gH_2$ are defined from $\gH$ by the sets $\vertices[\gH_1]$ and $\vertices[\gH_2]$ respectively.
\end{proof}

\begin{lemma}\label{lemma:impRoots}
	Let $\cf\colon \gH \to \gG$ is a even or odd skew fibration, with $\gG=\gG_1\gimp \gG_2$ and $\gG_1$ \mas.
	If there are two \mas $\gH'$ and $\gH''$ such that $\gH=\gH' \gimp \gH''$ and $\gH''$ cannot be written as $\gimp$ of two \mas,
	then  $\cf(v)\in \vertices[\gG_2]$ for all  $v\in \vertices[\gH'']$.
\end{lemma}
\begin{proof}
	Let $v\in \gH''$ such that $\cf(v)\in \gG_1$.
	Since $\cf$ preserves $\depsym$, then  $v\notin\irof[\gH]$.
	Thus $\gH''$ cannot be a single-vertex \ma.
	If $\gH''$ is a $\gsum$ of two \mas, then there is $z\in \irof[\gH'']$ such that $v\niedge^* z$, hence $v\conj z$ in $\gH$ but $\cf(v)\not\conj \cf(z)$ in $\gG$. Therefore $\cf$ is not an even skew fibration.
	Let $\cf(z)=w$. Then  $\cf(v)\disj w$ because $\cf(v)\in \gG_1$ and $w\in \irof[\gG]$.
	If there is a $u$ with $v\disj u$ in $\gH$ then there is $x\in \vertices[\gH]$ such that $u\iedge \iseven x$ and $v\iedge \iseven x$. 
	Since $x\iedge^* w$ we have $\cf(u)\disj w$, which means that $\cf$ cannot be an odd skew fibration either.  
	Then  $\gH''$ has to be of the shape $w\gmod\gH''_2$ and $\cf(w)\in \gG_2$ because $v\in \irof[\gH]$.
	We can conclude as for the previous case that $\cf$ is not an even or odd skew fibration.
	Contradiction.
\end{proof}

\begin{lemma}\label{lemma:arHomImp}
	Let $\cf\colon \gH \to \gG$ is a even or odd skew fibration, with $\gG=\gG_1\gimp \gG_2$ for an \ma $\gG_1$.
	If there is an \ma $\gH'$ such that $\gH=\gH' \gimp \gH''$,
	then there are  $\gH_1$ and $\gH_2$ such that $\gH=\gH_1\gimp \gH_2$
	and
	$\cf=\cf_1\gimp \cf_2$
	where 
	$\cf_1\colon \gH_1\to \gG_1$ and $\cf_2\colon \gH_2\to \gG_2$ are modal arena homomorphisms.
\end{lemma}
\begin{proof}
	By hypothesis, we can assume that $\gH$ is of the form $\gH=\gH'\gimp  \gH''$ where  $\gH''$ is not a $\gimp$ of two \mas.
	We conclude by Lemma~\ref{lemma:impRoots} that $\cf(v)\in \vertices[\gG_2]$ for any $v\in \vertices[\gH'']$. 
	If $\vertices[\gG_2]=\cf(\vertices[\gH''])$, then we conclude that $\gH_1=\gH'$ and $\gH_2=\gH''$.
	Otherwise, let $\gH'=\gH'_1\gsum \cdots \gsum \gH'_n$ such that $\gH'_i$ is a $\gsum$ of two \mas for no $i\in \intset1n$.
	If $v,w\in \vertices[\gH']$, then there is a $(\diedge\cup\dmedge)$-path from $v$ to $w$ in $\vertices[\gH']$ iff there is $i\in \intset1n$ such that $v,w\in \vertices[\gH'_i]$.
	Since $\irof[\gG]\subset\cf(\vertices[\gH''])$, this implies that if there is $i\in \intset1n$ such that $v,w\in \vertices[\gH'_i]$, then there is $(\diedge\cup\dmedge)$-path from $\cf(v)$ to $\cf(w)$ in $\vertices[\gG]\setminus\irof[\gG]$.
	That is, $\cf(\vertices[\gH'_i])$ is either a subset of $\vertices[\gG_1]$ or a subset of $\vertices[\gG_2]$ for all $i\in \intset1n$.
	Without loss of generality we assume there is $j$ such that that $\cf(\vertices[\gH'_i])\subset\vertices[\gG_1]$ for all $i\leq j$.	
	We conclude that $\gH_1=\gH'_1\gsum\cdots\gsum\gH'_j$ and $\gH_2=(\gH'_{j+1}\gsum\cdots\gsum\gH'_n)\gimp \gH''$.
\end{proof}

\begin{lemma}\label{lemma:arHomMod}
	Let $\cf\colon \gH \to \gG$ be a modal arena homomorphism  and  $\gG=\singlevertex[v] \gmod \gG'$.
	If $\cf$ is an even skew fibration then, $\gH=\singlevertex[w] \gmod \gH'$ and $\cf=\cid_w \gmod \cf'$ with $\cf'\colon \gH' \to \gG'$ an even skew fibration.
	If $\cf$ is odd skew fibration, then 
	\begin{itemize}
		\item 
		either $\gH=\singlevertex[w] \gmod \gH_2$ and $\cf=\cid_w \gmod \cf_2$ with $\cf_2\colon \gH_2 \to \gG_2$ an odd skew fibration;
		
		\item
		or $\gH=(w\gmod \gH_1 )\gsum \gH_2$ and 
		$\cf=\cpair{\cf_1}{{\cf_2}}$ with $\cf_1 \colon (w\gmod \gH_1) \to (v\gmod \gG_2 )$ and $\cf_2\colon \gH_2\to (v\gmod \gG_2)$.
	\end{itemize}
\end{lemma}
\begin{proof}
	If $\cf$ is an even skew fibration, then to conclude it suffices to remark there is a unique $w$ such that $\cf(w)=v$ since $v\in \irof[\gG]$.
	
	If $\cf$ is an odd skew fibration, let $w$ such that $\cf(w)=v$.
	If $\vertices[\gH]\setminus\set w=\mconeof{w}$, then we can conclude.
	Otherwise we conclude with  $\gH_2$ be the \ma with vertices in $\vertices[\gH]\setminus (\set w \cup \mconeof{w })$.
\end{proof}

In order to prove the converse of the previous theorem we give some additional definitions and useful lemmas.
\begin{definition}\label{def:skewOp}
	If $\cf_1:\gH_1\to\gG_1$ and $\cf_2:\gH_2\to\gG_2$ are modal arena homomorphisms, we define the following modal arena homomorphisms:
	$$
	\begin{array}{c@{\;=\quad}r@{\; : \;}c@{\to}cc}
		\cid_v 						& \id 				& \singlevertex[v] 			& \singlevertex[v]
		\\
		\cf_1 \gsum \cf_2		& \cf_1\cup \cf_2 & \gH_1 \gsum \gH_2 & \gG_1 \gsum \gG_2
		\\
		\cf_1 \gimp \cf_2		& \cf_1\cup \cf_2 & \gH_1 \gimp \gH_2 & \gG_1 \gimp \gG_2	& (\gH_1, \gG_1\neq \emptyset)	
		\\
		\cf_1 \gmod  \cf_2		& \cf_1\cup \cf_2 & \gH_1 \gmod \gH_2 & \gG_1 \gmod \gG_2
		\\
		\cpair{\cf_1}{\cf_2}	& \cf_1\cup \cf_2 & \gH_1 \gsum \gH_2 & \gG 						& (\gG_1 = \gG_2=\gG)
	\end{array}
	$$
\end{definition}

\begin{lemma}\label{thm:skewDec}
	Every (even) skew fibration is of the form
	$$
	\cid_\gG
	\qquad
	\iseven \cf \gsum \iseven \cg
	\qquad
	\isodd \cf \gimp \iseven \cg
	\qquad
	\cid_{v} \gmod \iseven \cg
	$$
	and every odd skew fibration is of the form
	$$
	\cid_\gG
	\qquad
	\cpair{\isodd \cf}{\isodd \cg}
	\qquad
	\isodd \cf \gsum \isodd \cg
	\qquad
	\iseven \cf \gimp \isodd \cg
	\qquad
	\cid_{v} \gmod \isodd \cg
	\qquad
	\cempty_{\gG}
	$$
	where $\iseven \cf $ and $\iseven \cg$ are even skew fibrations, 
	$\isodd \cf $ and $\isodd \cg$ are odd skew fibrations, $v\in \mvertices[\arof{H}]$, and $\gG$ can be any \ma.
\end{lemma}
\begin{proof}
	By case analysis, let $\cf \colon \gH \to \gG$ be a modal arena homomorphism, remarking that for any \ma $\gG$, the identity map $\cid_\gG$ is by definition an even and an odd skew fibration.
	
	If $ \iseven \cf \colon \gH \to \gG$ is  an even skew fibration, then 
	
	\begin{itemize}
		
		\item 
		if $\gG$ is a single-vertex \ma, then $\gH$ cannot be either of the shape $\gH_1\gsum \gH_2$ or $\gH_1\gmod \gH_2$ otherwise $\cf$ would not preserve $\conj$, or of the shape $\gH_1\gimp \gH_2$ otherwise it would not preserve $\depsym$.
		Then $\cf=\id_v$ with $v$ the unique vertex in $\vertices[\gH]=\vertices[\gG]$.
		
		\item 
		if $\gG=\gG_1\gsum \gG_2$, then by Lemma~\ref{lemma:arHomSum} we have that $\iseven \cf = \cf_1 \gsum\cf_2$ with  $\cf_1$ and $\cf_2$ arena homomorphisms. Since $\iseven \cf$ is an even  skew fibration, it follows by definition of $\gsum$ that $\cf_1$  and $\cf_2$ are even skew fibrations;
		
		\item 
		if $\gG=\gG_1\gimp \gG_2$, then 
		we define $V_1=\set{v\in \vertices[\gH]\mid \cf(v)\in \gG_1}$ and $V_2=\set{v\in \vertices[\gH]\mid \cf(v)\in \gG_2}$.
		We have that $V_2\neq\emptyset$ since $\cf$ preserve $\depsym$.
		If $V_1=\emptyset$, then $\cf=\cempty_{\gG_1}\gimp \cf_2$ with $\cf_2\colon \gH\to \gG_2$.
		Otherwise, $V_1\neq\emptyset$ and $\gH$ cannot be a single vertex.
		Similarly, $\gH$ cannot be of the shape $\gH_1\gsum\gH_2$ otherwise $\cf$ would not preserve $\conj$, nor of the shape 
		$v\gmod\gH_2$ otherwise $\cf$ would not be modal.
		We conclude by Lemma~\ref{lemma:arHomImp} that $\cf=\cf_1\gimp \cf_2$. 
		Moreover, since $\cf$ is an even skew fibration if follows that $\cf_2$ also preserves $\conj$ and satisfies skew lifting while $\cf_1$ preserve $\disj$ and satisfies odd skew lifting. 
		
		\item 	
		if $\gG=\singlevertex[v]\gmod \gG_2$, we conclude by Lemma~\ref{lemma:arHomMod}.
	\end{itemize}

	If $ \isodd  \cf \colon \gH \to \gG$ is  an odd skew fibration, then we proceed similarly.
	If $\gG$ is a single-vertex \ma, then $\gH$ cannot be of the shape $\gH_1\gmod \gH_2$ otherwise $\cf$ it would not be modal, or of the shape $\gH_1\gimp \gH_2$ otherwise it would not preserve $\depsym$.
	Let $\gH=\gH_1\gsum \gH_2$ such that $\gH_1\neq \gH_1'\gsum \gH_1''$. 
	Since $\isodd\cf$ preserve $\depsym$ and $\medge$, then $\gH_1$ is a single-vertex \mas.
	Moreover, $\cf_2\colon \gH_2\to \gG_2$ is an odd skew fibration by definition.
	Then $\cf=\cpair{\id_v}{\cf_2}$ with $v$ the unique vertex in $\vertices[\gH]=\vertices[\gG]$;
	
	If $\gG=\gG_1\gsum \gG_2$, $\gG=\gG_1\gimp \gG_2$ or $\gG=\singlevertex[v]\gmod \gG_2$ we apply a similar reasoning of the case of $\cf$ even skew fibration.
\end{proof}

\begin{theorem}\label{thm:skewfibToDer}
	Let $H$ and $H'$ be  \pformulas.
	If there is a skew fibration $\cf \colon \arof{H'} \to \arof{H}$, then $H' \provevia{\DOWN{\PIp}} H$.
\end{theorem}
\begin{proof}
	By Lemma~\ref{thm:skewDec} we can decompose any skew fibration using the operations in Definition~\ref{def:skewOp}.
	In particular, each $\cempty_\gG$  occurring in the decomposition corresponds to an application of a $\deep\wdrule$, $\deep\wtrule$ or $\deep\wirule$, while each occurrence of $\cpair{-}{-}$ corresponds to an application of a $\deep\crule$.
	We conclude by reconstructing a derivation in $\DOWN{\PIp}$ using this decomposition and the correspondence between \pformulas and \mas (\Cref{thm:Marena}).
\end{proof}

\section{Combinatorial proof}\label{sec:CP}

Using the results of the previous sections, we are able to define combinatorial proofs for the logics $\CK$ and $\CD$ and prove sound and completeness results for them.

\begin{definition}
	Let  $F$ be a formula and $\X\in \set{\CK,\CD}$. 
	
	An \emph{$\X$-intuitionistic combinatorial proof} (or \XICP) is a skew fibration $\cf \colon \cgG \to \arof F$
	from an $\X$-arena net $\cgG$ to the modal arena of a formula $F$ containing no occurrences of $\ldia \unit$.
	
	In particular, \emph{intuitionistic combinatorial proofs} (or \ICPs) from~\cite{ICP} are the special cases where no modalities occur, that is, an \ICP is a  skew fibration $\cf \colon \cgG \to \arof F$ from an arena net $\cgG$ to the arena of a modality-free formula $F$.
\end{definition}

\begin{theorem}\label{thm:ICPsoundandcomplete}
	Let $F$ be a formula and $\X\in\set{\CK,\CD}$. 
	Then 
	$$
	\provevia{\IpX} F \iff \mbox{there is an \XICP~} \cf \colon \cgG \to \arof F
	$$
\end{theorem}
\begin{proof}
	By \Cref{thm:decomposition} there is a \pformula $H$ such that $F=\spalletto H$ and
	$
	\provevia{\IpX} F  
	$ iff $
	\provevia{\PIMLLX} H' \provevia{\DOWN{\PIp}} H
	$
	for a  \pformula $H'$.
	By \Cref{thm:skewfibToDer} we have that 
	$H' \provevia{\DOWN{\PIp}} H$ iff there is a skew fibration  $\cf: \arof{\gH''} \to \arof H$.
	We conclude by  \Cref{thm:linSeq} since 
	by \Cref{thm:pol} we have
	$\provevia{\PIMLLK} H' $ iff $ \provevia{\IMLLX}  \spalletto{H'}$.
\end{proof}

\begin{lemma}
	Let  $\X\in \set{\CK,\CD}$.
	If $ \gH$ and $\gG$ are \ddags and $f\colon \vertices[\gH] \to \vertices[\gG]$, then it can be checked in polynomial time (in the size of $\gH\cup \gG$) if $f$ is a \XICP.
\end{lemma}
\begin{proof}
	All the following checks can be done in polynomial time: 
	that a \ddag $\cgG$ is an \ma; that an \ma is a $\K$- or $\D$-arena net; and that a map between two \mas is a skew fibration. 
\end{proof}

\begin{corollary}
	Let\/ $\X\in \set{\CK,\CD}$. Then  the \XICPs form a sound and complete proof system in the sense of Cook and Reckhow~\cite{cook:reckhow:79}.
\end{corollary}

\section{Winning Strategies}\label{sec:games}
\def\emptyseq{\epsilon}

\def\eventurn{$\evensym$-turn\xspace}
\def\oddturn{$\oddsym$-turn\xspace}
\def\eventurns{$\evensym$-turns\xspace}
\def\oddturns{$\oddsym$-turns\xspace}

\def\add#1{\mathsf{add}_{#1}}
\newcommand{\addi}[2]{\mathsf{add}_{#2}^{#1}}
\def\addsize#1{{\mclr h}_{#1}}
\def\emptyaddress{\epsilon}

\def\aproj#1{\lfloor\!\lfloor #1 \rfloor\!\rfloor}

In this section we provide the definition of winning strategies for a two-player game on a modal arena $\arof F$, and we show the correspondence between these strategies and $\CK$ and $\CD$ proofs of $F$.

\begin{definition}\label{def:winning}
	Let $\gG$ be a \ma. A \emph{move} is a vertex of $\gG$.
	Let $\jpath=\jpathi 0, \dots,\jpathi n$ be a sequence of distinct \emph{moves} (we denote by $\emptyseq$ the empty sequence).
	If $v$ and $w$ are two moves in $\jpath$, we say that a vertex $w$ \emph{justifies} $v$ whenever $v\iedge[\gG] w$.
	We call a move $\jpathi i$ in $\jpath$ a \emph{\evenmove} or \emph{\oddmove} if $i$ is respectively even or odd.
	We say that $\jpath$ is a \emph{view} if the following conditions are fulfilled:
	
	\hspace{-8pt}\begin{tabular}{ll}
		$\jpath$ is a \emph{play}: 
		&
		if $\jpath\neq\emptyseq$, then ${\jpathi{0}}\in\irof[\gG]$;
		\\		
		$\jpath$ is \emph{justified}:
		&
		if $i>0$, then $\jpathi i \iedge \jpathi j$ for some $j<i$;
		\\
		$\jpath$ is \emph{$\evensym$-shortsighted}:
		&
		if $ \iseven{\jpathi {i+1}}$ and  $\isodd{\jpathi {i}} $, then $\jpathi {i+1}\iedge\jpathi i$;
		\\
		$\jpath$ is \emph{$\oddsym$-uniform}:
		& 
		if $ \isodd{\jpathi {i+1}}$ and  $\iseven{\jpathi {i}} $, then $\lab{\jpathi {i+1}}=\lab{\jpathi {i}}$.
		\\
		$\jpath$ is \emph{modal}:
		& 
		$\jpathi i\in\avertices\cup\dvertices$.
	\end{tabular}	

	The \emph{predecessor} of a non-empty view $\jpath$ is the sequence obtained by removing the last move in $\jpath$.
	The \emph{successor} is the converse relation.
	A \emph{winning innocent strategy} (or \WIS) on $\gG$ is a finite predecessor-closed set $\strat$ of views in $\gG$ such that: 

	\begin{itemize}
		\item 
		$\strat$  and \emph{$\evensym$-complete}: if $\view\in\strat$ has even length, then every successor of $\view$ is in $\strat$;
		
		\item
		$\strat$ is \emph{deterministic} and  \emph{total}: if $\view\in\strat$ has odd length, then exactly one successors of $\view$ is in $\strat$;
		
		\item
		$\strat$ is \emph{$\ldia$-complete}: if $\iseven v\in\dvertices[\gG]$ occurs in $\strat$, then $\mconeof{v}\neq \emptyset$ and each $w$  such that $v\medge w$ occurs in $\strat$.
	\end{itemize}

	We say that a \WIS $\strat$ is \emph{atomic} if $\viewi i\in \avertices[\gG]$ for every $\view \in \strat$.
\end{definition}

\begin{remark}
	Our definition of \WIS restricted to (non-modal) arenas is the same  as the one in the literature, or simply a reformulation in our setting (see e.g. \cite{mur:ong:evolving} or \cite{ICP}).
\end{remark}
By means of example consider the strategy with maximal views shown in \Cref{fig:intro}.
We remark that the totality and $\evensym$-completeness of this strategy is guaranteed by the fact that the modal arena net is linked.

\def\Fview{framed abstract view\xspace}
\def\Fviews{framed abstract views\xspace}
\def\Aview{abstract view\xspace}
\def\Aviews{abstract views\xspace}

\begin{definition}

	Let $\X\in\set{\CK,\CD}$ and $\cgG$ be an $\X$-arena net. 
	A \emph{\Fview} of $\cgG$ is a reverse checked path 
	on $\linkgraph\cgG$.	

	We denote by $\aproj \jpath$ the sequence of moves in $\gG$ obtained by removing from a play $\jpath$ all modal vertices. 
	For example if $\jpath=\lbox u v  \lbox\ldia w$, then $\aproj{\jpath}=uvw$.
	
	An \emph{\Aview} $\aview$ in $\cgG$ is a sequence of atomic vertices in $\cgG$ defined as follows:
	\begin{itemize}
		\item either $\aview=\aproj{\view}$ for a \Fview $\view$ of $\cgG$;
		
		\item 
		or $\aview=\aproj{s_1} v_1 w_1  \aproj{s_3}\dots \aproj{s_{2k-1}} v_k w_k \aproj{s_{2k+1}}  $
		for a \Fview $\view= s_1 v_1 s_2 w_1 s_3\dots s_{2k-1} v_k s_{2k} w_k p_{2k+1}  $  of $\cgG$ 
		with  $v_i,w_i\in\dvertices$ such that $v_i\axlink w_i$ for all $i\in\intset0k$;
	\end{itemize}	
	Note that by definition, an \Aview of a non-modal arena net is a reverse path in $\cgG$.
\end{definition}

We recall the result on \ICPs from \cite{ICP} which we aim to extend \CKICPs and \CDICPs in this section.
\begin{theorem}[\cite{ICP}]\label{thm:winningI}
	If $\cf: \cgG \to \arof F$ is an \ICP of a modality-free formula $F$, 
	then the set of images of all \Aviews of $\cgG$ is a \WIS on $\arof F$.
\end{theorem}	
For this purpose, we define \emph{frames} as equivalence classes of modal vertices in the arena induced by the views.
They are meant to reconstruct the information about the applications of modal axioms, that are, the $\axeq$-equivalence classes of the modal arena net of the \ICP.
This information allows us to ``de-contract'' the formula $F$ in such a way to obtain a formula $F'$ which admits a linear derivation.

Let $\gG=\arof F$ be a \ma.
The  \emph{address}  of a vertex in $v\in\vertices[\gG]$ is 
the unique (possible empty) sequence of modal vertices $\add v=m_1 \cdots  m_k$  such that 
$m_0=v$ and $m_{i}=\pmv{m_{i-1}}\neq m_{i-1}$ for each $i\in \intset1{k}$.
Intuitively, the address of a vertex $v$ is the list of the modalities in the path the node corresponding to $v$ to the  root of the formula tree $\ftree F$.
We denote by $\addsize{v}=\sizeof{\add v}$ and  
$\addi h {v}$ the $h^{\mbox{\small th}}$ element $m_h$ in $\add v$. 
If $\view $ is a view, we write  $\addsize{\view}=\max\set{\addsize{v}\mid v\in\view}$.
Moreover, if  $\strat$  is a strategy on $\gG$, we say that $v\in\vertices[\gG]$ is \emph{involved}  in $\strat$ if either $v\in \view$ or if $v\in\add{\viewi i}$ for a view $\view \in \strat$.

\begin{definition}\label{def:winningCK}
	Let $\view=\viewi 1 \cdots \viewi n$ be a view on a \ma $\gG$.
	
	We say that $\view$ is \emph{well-framed} if $\sizeof{\add{\viewi {2k}}}=\sizeof{\add{\viewi {2k+1}}}$ for every even $2k\in \intset{0}{n-1}$.
	A strategy is \emph{well-framed} if each view in it is.

	If $\view$ is well-framed, then we define its \emph{framed view} as 
	the $\addsize{\view}\times n$  matrix 
	$\fviewof{\view}=\big( \fviewof{\view}_{0}, \dots, 	\fviewof{\view}_{n}\big)$ 
	with elements in $\vertices[\gG]\cup\set{\emptyseq}$ such that  
	each column $\fviewof{\view}_{i}$ is defined as follows:
	{\small$$
	\fviewof{\view}_{i}
	=
	\begin{pmatrix}
		\fviewof{\view}^{\addsize{\view }}_{i}	&=&  \addi{\addsize{\viewi{i}}}{\viewi {i}} 
		\\
		&\vdots
		\\
		\fviewof{\view}^{h_i+1}_{i}&=&\addi{1}{\viewi {i}} 
		\\
		\fviewof{\view}^{h_i}_{i}&=&\emptyseq
		\\
		&\vdots		
		\\
		\fviewof{\view}^1_{i}&=&\emptyseq
		\\
		\fviewof{\view}^0_{i}&=& \viewi i
	\end{pmatrix}
	$$}
	where a $h_i\in\intset0{\addsize{\view }}$ defined for each $i\in\intset0n$.
	
	Moreover, each $\fviewof{\view}$ induces an equivalence relation $\strateq{\gG}{\view}$ over $\vertices[\gG]$ given by the symmetric, transitive, and reflexive closure of the following relations:
	$$
	u\strateqone{\gG}{\view}w
	\quad \mbox{ iff }
	\quad
	\begin{array}{c}
		\mbox{$u=\fviewof{\view}^h_{2k}$ and $w=\fviewof{\view}^h_{2k+1}$}
		\\
		\mbox{for a $2k < n$ and a $h\leq{\addsize{\view}}$}
	\end{array}
	$$

\end{definition}

We write $\add v\strateq{\gG}{\view} \add w$ if $v$ and $w$ are involved in $\view$ and $\addi{k}{v}\strateq{\gG}{\view} \addi kw$ for all $k$.

\begin{lemma}\label{lemma:functo}
	Let $\gG$ be a \ma, $\view$ be a well-framed
	view on $\gG $, and $v,w\in\vertices[\gG]$.
	If $v\strateq{\cgG}{\view}w $,
	then 
	there are some $i,j\in\intset0n$ and a $k\in\intset{0}{\addsize{\view}}$ such that 
	$v = \fviewof{\view}^k_{i}$ and $w = \fviewof{\view}^k_{j}$.
	Moreover, for any $h>k$ we have
	$\fviewof{\view}^{h}_{i}\strateq{\cgG}{\view}\fviewof{\view}^{h}_{j}$.
\end{lemma}
\begin{proof}
	Let us write $\axeq$ instead of $\strateq{\cgG}{\view}$.
	If $v\axeq w $, then by definition there are $i,j,k\in \N$ such that 	$v = \fviewof{\jpath}^k_{i}$ and $w = \fviewof{\jpath}^k_{j}$.
	To prove that $\fviewof{\view}^{h}_{i}\axeq\fviewof{\view}^{h}_{j}$ for all $h\geq k$, we assume w.l.o.g. that $j\geq i$ and we proceed by induction on $n=j-i$. 
	If $n=0$, then the statement trivially holds since $i=j$ and $\axeq$ is reflexive.
	If $n>0$, we make case analysis on the parity of $j$.
	If $j$ is odd, then $ \fviewof{\jpath}^{h}_{j} 		\axeq 		 \fviewof{\jpath}^{h}_{j-1}$
	for all $h\geq k$ by definition of $\axeq$. 
	By transitivity  of $\axeq$  we have
	$\fviewof{\jpath}^{h}_{j-1}\axeq\fviewof{\jpath}^{h}_{i}$.
	We conclude by inductive hypothesis since $(j-1)-i<n$.
	%
	If $j$ is even, then 
	$ \fviewof{\jpath}^{k}_{j}  	\axeq 		\fviewof{\jpath}^{k}_{j-1} $
	if and only if
	either 
	$ \fviewof{\jpath}^{k}_{j} 					= 						\fviewof{\jpath}^{k}_{j-1}$ 
	or
	$ \fviewof{\jpath}^{k}_{j}  	\axeq 		\fviewof{\jpath}^{k}_{j+m} $
	for a $m>0$ such that $\fviewof{\jpath}^{k}_{j+m}= \fviewof{\jpath}^{k}_{j'}  $ for a $j'<j$.
	In the first case we conclude by inductive hypothesis since $ \fviewof{\jpath}^{h}_{j} = \fviewof{\jpath}^{h}_{j-1}$ for all $h>k$, and therefore 
	$ \fviewof{\jpath}^{h}_{j} \axeq \fviewof{\jpath}^{h}_{j-1}$.
	In the second case we conclude by inductive hypothesis since $j'<j$.	
\end{proof}

\begin{definition}
	Let $\strat$  be a well-framed strategy on a \ma $\gG$.
	We say that $\strat$ is \emph{linked} if
	for every $\view\in\strat$ the $\strateq{\gG}{\view}$-classes are of the shape $\set{\isodd v_1, \dots, \isodd v_n, \iseven w}$.
	This induces an edge-relation $u\stratlink{\cgG}{\strat}w=\set{\isodd u\strateq{\cgG}{\view} \iseven w \mid \view \in \strat}$.
	
	 A \emph{$\CK$-framed} strategy on a \ma $\gG$ is a well-framed linked strategy $\strat$ such that 
	for each $\iseven w\in\mvertices[\gG]$ involved in $\strat$ the following conditions are fulfilled:
	\begin{enumerate}  
		
		\item\label{cond:frames:CKbox} 
		if $w\in \bvertices[\gG]$, then $v\in \bvertices[\gG]$ for any $v\stratlink{\gG}{\strat} w$;

		\item\label{cond:frames:CKdia} 
		if $w\in \dvertices[\gG]$, then 
		$v\in \dvertices[\gG]$ for a unique $v$ such that $v\stratlink{\gG}{\strat} w$.
		
	\end{enumerate}
	 A \emph{$\CD$-framed} strategy on a \ma $\gG$ is a well-framed linked atomic 
	 strategy $\strat$ such that 
	 for each $\iseven w\in\mvertices[\gG]$ involved in $\strat$, it satisfies Condition \ref{cond:frames:CKbox} plus the following
	\begin{enumerate}  
		\setcounter{enumi}{2}	
		
		\item\label{cond:frames:CDdia} 
		if $w\in \dvertices[\gG]$, then  $v\in \dvertices[\gG]$ for at most one $v$ such that $v\stratlink{\gG}{\strat} w$.

	\end{enumerate} 
For $\X\in\set{\CK,\CD}$, we say that $\strat$ is a \XWIS if it is a $\X$-framed \WIS.
\end{definition}

\begin{example}
	Let us consider the two non $\CK$-provable formulas 
	$F=\lbox a_1 \imp a_0$ and  $F'=(\lbox a_2 \imp \lbox b_1) \imp \lbox (a_3 \imp  b_0) $ 
	where we enumerate occurrences of the same atom to improve readability.
	
	The unique view on $\arof F$ is $\iseven a_0 \isodd a_1$. 
	Since  $\add{a_0}=\emptyseq $ and $\sizeof{ \add{a_1}}=1$, we conclude that any strategy on $\arof F$ is not be well-framed.
	
	Similarly, the unique maximal view on $\arof{F'}$ is $\iseven b_0 \isodd b_1 \iseven a_2 \isodd a_3$. 
	This view is well-framed. However  its frame contains all three modalities of the formula, two of which are $\evensym$; Hence any  strategy on $\arof{F'}$ would not be $\CK$-framed
\end{example}

As consequence of Lemma~\ref{lemma:functo} we have the following result

\begin{corollary}[Functoriality]\label{cor:functoriality}
	Let $\strat$ is a well-framed strategy on a \ma $\gG$. 
	If $v,w\in\vertices[\gG]$ and $v\strateq{\gG}{\view} w$, then $\add v\strateq{\gG}{\view} \add w$.
\end{corollary}

The rest of this section is devoted to show how to use \XICP to expose the correspondence between \XWISs and proof in $\CX$ for $\X\in\set{\CK,\CD}$.
Since \XICPs are sound and complete (see \Cref{thm:ICPsoundandcomplete}), it is easy to show that we can associate to any proof a \XWIS using the following lemma:

\begin{lemma}\label{lemma:CPtoWS}
	Let $\X\in\set{\CK,\CD}$.
	If $\cf: \cgG \to \arof F$ is a \XICP of a formula $F$, then the image by $\cf$ of all \Aviews of $\cgG$ is a \XWIS on $\arof F$.
\end{lemma}
\begin{proof}
	The image by $\cf$ of an \Aview is a play and it is $\evensym$-shortsighted since $\cf$ preserves $\depsym$ and if $v,w\in\avertices[\cgG]$, then $v\iedge w$ in $\linkgraph{\cgG}$ only if $v\oedge w$.
	Moreover, in a modal arena if  $v\axeq[\cgG] w$  and $v,w\in\avertices[\cgG]$, then $\lab v=\lab w$.
	We deduce that $\cf$ is $\oddsym$-uniform since $\cf$ also preserves $\labsym$.
	Hence the image by $\cf$ of an abstract views on $\cgG$ is a view on $\arof F$.
	
	Since for any abstract view $\view$ on $\cgG$ we have that $\viewi{2k+1}=v_{2k+1}\axeq[\cgG] v_{2k}=\viewi{2k}$, then
	by functoriality  of $\cgG$ (Condition \ref{cond:functorial} in Definition~\ref{def:KarenaNet}),  we have $\pmv{v}_{2k+1}\axeq[\cgG] \pmv{v}_{2k}$. 
	This allows us to conclude by induction that $\addsize{v_{2k}}=\addsize{v_{2k+1}}$  in $\cgG$, i.e., $\view$ is well-framed view since since $\cf$ is modal and preserves $\medge$.
	
	The $\evensym$-completeness follows by definition of $\oedge \cup \omedge$.
	Determinism of the strategy follows by the fact that $\linkgraph{\cgG}$ is $\X$-correct, 
	then 
	\begin{itemize}
		\item if $\X=\CK$, then for every $\iseven w\in\avertices\cup\dvertices$ there is a unique vertex\footnote{Observe that this is not true for $\lbox$-vertices.} $\isodd v$ such that $\isodd v \axlink \iseven w$.
		Moreover in this case $\ldia$-completeness follows the non-empty modalities Conditions~\ref{cond:mod};
		
		\item and if $\X=\CD$, then $\strat$ is atomic and atomic vertices are paired in $\axeq$-classes.
		Moreover in this case $\ldia$-completeness is valid since no $\ldia$ occurs in $\strat$.
	\end{itemize}
	
	We conclude since by definition $\linkgraph{\cgG}$ is linked and $\X$-correct; thereby $\strat$ is $\X$-framed.
\end{proof}

To prove that each \XWIS correspond to a proof in $\CX$, we give a procedure to 
define an \XICP $\cf \colon \cgG \to \arof F$ using the information provided by the arena $\arof F$ and the strategy $\strat$.
Using the property of being well-framed, we are able to reconstruct some paths on $\arof F$ which should be the images by the skew fibration $\cf$ of the \Fviews in the modal arena net $\cgG$.

\begin{definition}
	Let $\view$ be a well-framed view on a \ma $\gG$ of length $n$ .
	We define the \emph{pre-view} of $\view$ as the sequence of vertices in $\gG$
	$$
		\begin{array}{ll}
			\aview=\aview_0,\aview_{2},\aview_{4}, \dots, \aview_{2k}
			&
			\mbox{if $n$ is even}
			\\
			\aview=\aview_0,\aview_{2},\aview_{4}, \dots, \aview_{2k}, \aview_{2k+1}
			& 
			\mbox{if $n$ is odd}
		\end{array}
	$$
	where for all $i\in\intset1k$ we have 
	$$
	\begin{array}{l@{= \quad}l}
		\aview_0
		&
		\fviewof{\view}^{\addsize{\view}}_{0}, \dots, \fviewof{\view}^0_{0}
		\\
		\aview_{2i}
		&
		\fviewof{\view}_{2i-1}^0
		\dots, 
		\fviewof{\view}_{2i-1}^{h_i}
		\fviewof{\view}_{2i}^{h_i}
		\dots 
		\fviewof{\view}_{2i}^{0}
		\\
		\aview_{2k+1}
		&
		\fviewof{\view}_{2k+1}^{0}, \dots,   \fviewof{\view}_{2k+1}^{\addsize{\view}}
	\end{array}
	$$
	for a  
	$h_i=\max \set{ h \mid \fviewof{\view}^h_{2i-1}\strateq{\gG}{\strat}  \fviewof{\view}^h_{2i} }$
	if $ \add{\viewi{2i+1}}  \nstrateq{\gG}{\strat}  \add{\viewi{2i}}$ and $h_i=0 $ otherwise.
	We denote by $\astrat$ the set of the pre-views of all the views in $\strat$, that is, 
	$\astrat=\set{\aview\mid \view\in\strat}$.
	
	A \emph{unchecked prefix} of a $\aview\in\astrat$ is a sequence of vertices $\unchecked s$ obtained by replacing in a prefix $s$ of $\aview$ each subsequence of the form $vrw$ with $vw$ whenever $v \strateq{\gG}{\strat} w$.

\end{definition}
\begin{example}
	Let us consider the maximal views on the modal arena of $\lbox_1 (b_1\imp b_0)\imp a_1) \imp( \ldia_1 c\imp \ldia_0(a_2\land a_0))$ from \Cref{fig:intro}.
	From the leftmost and central views we respectively define  the pre-views 
	$\iseven \ldia_0 \iseven a_0 \isodd a_1 \iseven b_0 \isodd b_1 \isodd \lbox_1$
	and
	$\iseven \ldia_0 \iseven a_2 \isodd a_1 \iseven b_0 \isodd b_1 \isodd \lbox_1$.
	In particular, the sequence $\ldia_0\lbox_1$ is the unique unchecked prefix associated to these two sequences.
\end{example}

Given an arena $\arof F$ and a \XWIS $\strat$ and following this intuition, 
we reconstruct a partitioned modal arena $\cgG_\strat$ and a map $\cf_\strat$ from its vertices to the ones in $\arof F$ as follows.

\begin{definition}\label{def:arofStrat}
	Let $\X\in\set{\CK,\CD}$ and $\strat$ be a \XWIS on $\arof F$.
	We define 
	an arena $\cgG_\strat=\tuple{\vertices[\cgG], \iedge[\cgG], \medge[\cgG], \axeq[\cgG]}$ 
	and 
	a map $\cf_\strat=\cf : \vertices[\cgG] \to \arof{F}$ as follows:
	\begin{itemize}

		\item 
		in $\vertices[\cgG]$ there is one vertex for each non-empty unchecked prefix $\unchecked s$ of a pre-view $\aview\in\astrat$ 
		(whose label is the same of the last vertex in $\unchecked{s}$). 
		That is, 
		\begin{equation}\label{decont:Lvertices}
			\hspace{-10pt}
			V_{\cgG}=
				\set{v_s \mid s \mbox{ is a non-empty {unchecked} prefix of a } \aview\in\astrat}
		\end{equation}
		\begin{equation}
		\lab{v_{s'w}}=\lab{w}
		\end{equation}
		
		\item 
		by definition every vertex is of the form $v_s$ for a non-empty sequence $s$ of vertices in $\vertices[\arof F]$.
		We define the map $\cf:\vertices[\cgG]\to\vertices[\arof F]$ in such a way it maps each $v_s\in \vertices[\cgG]$ to the last vertex of $s=s'w$. That is,
		\begin{equation}\label{decont:skew}
			\cf(v_{s'w})=w
		\end{equation}

		\item 
		there is an edge $v\iedge[\cgG] w$ 
		whenever
		$\cf(v)\iedge[\arof F] \cf(w)$,
		and 
		the images of $v$ and $w$ occur in the addresses or are respectively some vertices $x$ and $y$
		such that
		either $\iseven x$ and $\isodd y$ occur in a same view in $\strat$, 
		or 
		there is $s\in \strat$ such that $s\iseven y$ and $su\isodd v$ occur in $\strat$.
		That is, 

		\begin{equation}\label{decont:iedge}
			\iedge[\cgG]=
			\Set{v\iedge w
				\quad
				\begin{array}{|@{\quad}l}
					\cf(v)\iedge \cf(w) 
					\mbox{ and 
						there are $x,y \in\vertices[\arof F]$ 
						such that 
					}
				\\
					\mbox{$\cf(v)=x$ or $\cf(v)\in \add x$,
						$\cf(w)=y$ or  $\cf(w)\in \add y$
					}
				\\
					\mbox{
						and either $s\isodd  y\iseven x$, or both $s\iseven y$ and $su\isodd x$ are in $\strat$
					}
				\end{array}
			}
		\end{equation}

		\item 
		there is an edge $v\medge[\cgG] w$ 
		whenever
		$\cf(v)\medge[\arof F] \cf(w)$,
		and 
		$v$ and $w$ occur in a same pre-view in $\astrat$.
		That is, 	
			\begin{equation}\label{decont:medge}
				\hspace{-10pt}
				\medge[\cgG]=
				\set{v\medge w
					\mid
					\cf(v)\medge[\arof F]\cf(w) 
					\mbox{ and }
					\cf(v),\cf(w)\in  \aview \mbox{ for a } \view\in\strat
				}
			\end{equation}

		\item
		we define $v\axeq[\cgG] w$ as the symmetric and transitive closure of 
		the edge-relation $\stratlink{\arof{F}}{\strat}$. That is,
		\begin{equation}\label{decont:axeq}
			\hspace{-10pt}
			\axeq[\cgG] =
				\set{v\axlink w \mid \cf(v)\dstratlink{\arof{F}}{\view}^*\cf(w)  \mbox{ for a } \view\in\strat}
		\end{equation}

	\end{itemize}
\end{definition}

\begin{remark}\label{rem:linked}
	By definition  $\strateq{\arof{F}}{\strat}=\bigcup_{\view\in\strat} \strateq{\arof{F}}{\view}$ is not an equivalence relation over $\arof F$.
	In fact in $\arof F$ we may have some vertices $u$,$v$ and $w$ such that  $u \stratlink{\arof{F}}{\strat} v$ and $u \stratlink{\arof{F}}{\strat}w$ and $v \nstrateq{\arof F}{\strat} w$.
	
	If we additionally assume that $\strat$ is linked, then we conclude that $u \strateq{\arof{F}}{\view_1} v$ and $u \strateq{\arof{F}}{\view_2}w$ for two distinct $\view_1, \view_2\in \strat$. Hence, the vertex $u$ in $\arof F$ admits at least two different pre-images in $\gG$.
	Then we conclude, as the homonymy suggests,  that $\cgG_\strat$  is a linked modal arena.
\end{remark}

\begin{remark}\label{rem:lifting}
	By definition of $\medge$, we have that 
	$v\omedge[\cgG_\strat]w$
	iff 
	$v\dmedge[\cgG_\strat]w \mbox{ and } v=\aviewi i \mbox{ and } w=\aviewi j \mbox{ for a  } i>j  $.
	That is, any pre-view is a reverse cautious path on $\linkgraph{\cgG_\strat}$, that is, a framed abstract view.
	It follows that the \Aview which can be extract from a pre-view $\aview$ of a $\view\in\strat$ is exactly the view $\view$.
	In other words, the function mapping a view in its pre-view is the left adjoint of the function mapping a \Fview to its associated \Aview.
\end{remark}

Hence, by proving that $\cf_\strat$ is an \XICPs we can prove that we can associate a proof in $\CX$ to any \XWIS.

\begin{restatable}{lemma}{lemmaWStoCP}\label{lemma:WStoCP}
	Let $\X\in\set{\CK,\CD}$
	If $F$ is a formula and $\strat$ a \XWIS on $\arof F$, 
	then there is a \XICP $\cf\colon \cgG \to \arof F$.
\end{restatable}
\begin{proof}
We only prove the result for $\CK$ since the proof for $\CD$ is similar but easier since \CDWISs are atomic.

We use Definition~\ref{def:arofStrat} to define an \XICP $\cf_\strat\colon \cgG_\strat \to \arof F$ form $\strat$ and $\arof F$.
That is, we prove that 
the map $\cf_\strat$ 
and
the modal arena $ \cgG_\strat$
defined in Definition~\ref{def:arofStrat}
are respectively a skew fibration
and, whenever $\strat$ is $\CX$-framed,
an $\X$-arena net.

	The arena $\cgG$ is linked by definition of $\axeq[\cgG]$ as remarked in Remark~\ref{rem:linked}.
	To conclude that $\cgG$ is a $\CK$-arena net we have to check the following conditions:
	\begin{enumerate}
		\item $\linkgraph{\cgG}$ is acyclic: 
		if a checked  path contains a cycle, then we can define a framed abstract view for any number of iterations of this cycle.
		Then $\strat$ should contains infinite views corresponding of the image through $\cf$ of infinite abstract views on $\cgG$. Absurd.
		
		\item $\linkgraph{\cgG}$ is functional: 
		for any $\isodd{\aviewi i}$   there is a $k\leq i$  such that 
		$\aviewi k=\viewi h$ occurs in $\view$ and either $k=i$ or $\aviewi i \medge[\cgG] \aviewi k$.
		Since $\view$ is justified, then there is $l<h$ such that $\viewi h \iedge \viewi l$. Then there is $j<i$ such that $\aviewi j=\viewi l$.
		By the fact that $\medge$ is modal (see Definition~\ref{def:ma}), we conclude that $\aviewi i \iedge[\gG] \aviewi j$.

		\item $\linkgraph{\cgG}$ is functorial: it follows Corollary~\ref{cor:functoriality};
		
		\item $\linkgraph{\cgG}$ has almost all non-empty modalities: 
		let $ v\in \mvertices[\cgG]$ such that $v=v_s$ for a prefix $s=s'\fviewof{\view}_{2k}^{h}$ of a $\aview\in\astrat$.
		If $v\in\bvertices$, then $h>0$ (since no $\lbox$ occurs in a \Aview) and there is $w=w_{s''}$ such that $v\medge[\cgG]w$
		such that 
		either $s''={s'\fviewof{\view}_{2k}^{h}\fviewof{\view}_ {2k}^{h-1}}$ if $\iseven v$,
		or $s''=s'$ if $\isodd v\in\bvertices[\cgG]$.
		If $\iseven v\in\dvertices$, then we conclude by $\ldia$-completeness.

		\item $\linkgraph{\cgG}$ is $\CK$-correct: 
		it follows from the fact that $\strat$ is $\CK$-framed and that, by definition, $\axlink[\cgG]=\stratlink{\cgG}{\strat}$. 
		
	\end{enumerate}

	The map $\cf$ is a skew fibration we have to check the following conditions:
	\begin{itemize}
		
		\item  $\cf$ preserves $\ell$, $\iedge$, and $\medge$: by definition;
		
		\item  $\cf$ preserves $\depsym$: since $\cf$ preserves $\iedge[\cgG]$, then $\dep v \geq \dep {\cf(v)}$.
		If $\dep v > \dep {\cf(v)}$ then there should be a $w$ such that $\cf(v)\iedge w$ and $\dep{w}\geq \dep{\cf(v)}$ which by  Lemma \ref{lemma:cones},  implies that $\arof F$ is not L-free. Contradiction;
		
		\item $\cf$ is modal:  if $\cf(v)\medge \cf(w)$, then by definition there is a $k$ such that $\addi k{\cf(w)}=\cf(v)$. We conclude by letting $v'=\addi k{w}$.
		
		\item $\cf$ preserves $\conj$: it follows from the fact that $\cf$ preserves $\iedge$ and $\depsym$;

		\item $\cf$ has the skew lifting property:
		we let $w\in\vertices[\arof F]$ such that $w\conj \cf(v)$ for a $v\in\vertices[\cgG]$ and we prove  that there is always a $u$ such that $v \conj u$ and $\cf(u)\nconj w$.

		If there is no meeting point of $w$ and $\cf(v)$, then we conclude by letting $u\in\irof[\cgG]$ such that $u\conj v$ and $w\iedge^* \cf(u)$.
		
		Otherwise, we let $\isodd x$ (hence $x\notin\irof[\arof F]$) be a meeting point of $w$ and $\cf(v)$. 
		By Lemma~\ref{lemma:modalities} can assume w.l.o.g. that $x\in\avertices[\arof F]$.
		Moreover, we can also assume that $x$ is in the image of $\cf$.
		In fact, since the meeting point exists, then there is a $\iseven r$ (at least one $r\in\irof[\arof F]$) such that $w\iedge^n r$ and $\cf(v)\iedge^m r$; we can assume $r\in\avertices$ and by determinism of $\strat$ we have a $z\in\avertices$ in the image of $\cf$ such that $z\iedge r$; 
		thus by Lemma~\ref{lemma:cones} 
		either $z$ is the meeting point, 
		or for all $r'\in \avertices$ such that $r'\iedge z$, $r'$ is in the image of $\cf$ since $\strat$ is total and $\evensym$-complete; we conclude by induction.
		
		We can deduce that $sx\in\strat$ for a $s\in\strat$.
		We now let $y$ such that $w\iedge^*y\iedge x$.
		Since $sx\in\strat$ and $\isodd x$, then by $\evensym$-completeness we have $sxy\in\strat$ for every $y\in \avertices$ such that $y\iedge x$; thus $\cf(u)=y$ for a $u=v_{sxy}\in\vertices[\cgG]$.
		We conclude since the meeting point of $w$ and $\cf(u)$ is $\cf(u)=\iseven{y}$ and the meeting point of $\cf(u)$ and $\cf(v)$ is $\isodd x$.
	\end{itemize}\qedhere

\end{proof}

We are able to prove a soundness and completeness result for \XWISs.

\begin{theorem}\label{thm:winningCK}
	Let $F$ be a formula and $\X\in\set{\CK,\CD}$.
		$$
		F \mbox{ is provable in } \CX
		\iff
		\mbox{ there is a \XWIS on $\arof F$.}
		$$
\end{theorem}
\begin{proof}
	By \Cref{thm:ICPsoundandcomplete} we know that \XICPs are a sound and complete proof system for $\IpX$.
	We conclude the proof using  Lemmas~\ref{lemma:CPtoWS} and \ref{lemma:WStoCP}  which state the correspondence between \XICPs of a formula $F$ and \XWIS on $\arof F$.
\end{proof}

\section{Conclusions and Future Works}\label{sec:conclusion}

In this paper we present two semantics for proofs of the disjunction-free and unit-free fragment of the constructive modal logics $\CK$ and $\CD$.

The first semantics is given by extending the syntax of  \ICPs from \cite{ICP} by reshaping some techniques from the previous work on combinatorial proofs for modal logic \cite{acc:str:CPK} to fit with the syntax required to capture intuitionistic logic.
We define \mas which extend the syntax of a Hyland-Ong arena~\cite{mur:ong:evolving} 
in order to represent modal formulas by finite directed graphs, and we define modal arena nets which are \mas equipped with a vertex partition capturing axioms in $\CK$ and $\CD$.
Then we prove that skew fibrations from a modal arena net to the arena of a formula are sound and complete with respect to the logics $\CK$ and $\CD$.

The second semantics is given in terms of winning innocent strategies over modal arenas.
It has been designed by extending the relation between \ICPs and winning strategies shown in \cite{ICP}: 
the set of paths in the linking graph of the arena net of the  \ICP is mapped by the skew fibration to a winning innocent strategy on  the formula arena.
This relation has been further refined by showing that for $\CK$ and $\CD$ it is possible to restrict this set of paths to specific ones passing on atoms and diamonds only.

We get the following result for our two new semantics:
\begin{theorem}[Full completeness]
	Let $F$ be a formula and $\X\in\set{\CK,\CD}$. Then
	\begin{enumerate}
		\item\label{full:LtoCP} There is a surjection from the set of factorised proofs of $F$  and  the set  \XICPs  of $F$. 
		
		\item\label{full:CPtoWS} There is a surjection from the set of \XICPs  of $F$ and the set of  \XWISs on  $\arof F$.
		
		\item There is a surjection from the set of $\IpX$-derivations of $F$ and the set of  \XWISs on $\arof F$.
	\end{enumerate}
\end{theorem}
\begin{proof}
	\begin{enumerate}
		\item
		The proofs of \Cref{thm:linSeq} and \Cref{thm:skewfibToDer} allow to establish full maps respectively from $\IMLLX$-derivations to $\X$-arena nets, and 
		from $\DOWN\PIp$-derivations to skew fibrations.
		We conclude by composing these maps.

		\item 
		The proof of Lemma~\ref{lemma:CPtoWS} establishes a map from the set of \XICPs of $F$ to the set of  \XWISs on $\arof F$. The proof of Lemma~\ref{lemma:WStoCP}  associates an \XICP to an \XWIS $\strat$. 
		As remarked in Remark~\ref{rem:lifting}, the image by $\cf_\strat$ of the \Aviews on the linking graph of the modal arena net $\cgG_\strat$ defined in Definition~\ref{def:arofStrat} is exactly initial \XWIS $\strat$.
		We conclude since every \XWIS $\strat$ on $\arof F$ is the image by $\cf$ of the \Fviews in the \XICP $\cf_\strat \colon \cgG_\strat \to \arof F$.
		
		\item Direct consequence of \ref{full:LtoCP} and \ref{full:CPtoWS}.\qedhere
	\end{enumerate}
\end{proof}

We conclude by presenting some lines of inquiry that have been initiated by the content of this paper.

\textbf{Game semantics for $\CK$ and $\CD$.}

We are currently investigating the compositionality of \CKWISs and \CDWISs in order to define the game semantics for these logics.
It seems natural that the standard definitions of canonical strategies of game semantics (e.g., copy-cat, projections and evaluation)
can be employed in our framework.
However, the additional condition on frames requires a careful investigation which goes beyond the scope of this paper.


\textbf{Relation between $\lambda$-terms and winning strategies.}
For propositional intuitionistic logic the relation between $\lambda$-terms and \WISs is well-known~\cite{hyl:ong:PCF,Dan:Her:Regn:gamesandmachine,Hughes:97:games}
The exact correspondence between our
\WISs and 
$\lambda$-terms for $\CK$ and $\CD$ \cite{bie:paiva:intuitionistic} 
is under investigation, but out of the scope of this paper.


\textbf{Proof equivalence in constructive modal logics.}

Both \ICPs and \WISs induce a proof equivalence between proofs defined as 
``two derivations are equivalent iff they are represented by the same semantic object''.

We conjecture that, as proven in \cite{ICP} for the intuitionistic combinatorial proofs for the logic $\Ip$, the combinatorial proofs presented in this paper
capture the proof equivalence defined on sequent calculus by 
independent rules permutations,  weakening/contraction-comonad, and  excising, i.e., the permutation of weakening which removes subproofs shown below (see the permutation in the bottom-right corner of \Cref{fig:permutations}).
$$
\vlderivation{
	\vliin{}{\limprule}{\Gamma, A \imp B, \Delta \vdash C}
	{\vlpr{\dD}{}{\Gamma\vdash  A }}
	{\vlin{}{\Wrule}{B,\Delta \vdash C}
		{\vlhy{\Delta, \vdash C}}
}}
\quad
\oset{\mathsf e}{\rightsquigarrow}
\quad
\vlderivation{
	\vliq{}{\Wrule}{\Gamma, A \imp B, \Delta \vdash C}
	{\vlhy{\Delta, \vdash C}}
}
$$ 
We also conjecture that the full completeness results can be stated with respect to all proofs of a formula, and not only the factorised ones.

However, in the presence of modalities, the proof of these results is
much more involved (see \Cref{fig:permutations}), and would go beyond the scope of this
paper. 
Although, it is easy to see that whenever two sequent proofs
are equivalent modulo rule permutations, they are mapped to the same
combinatorial proof, the converse is far from trivial, in particular,
it is not true in the classical case.

Moreover, an additional problem seems to arise for $\CK$ which is similar to the well-known ``jump-problem'' for multiplicative linear logic  proof nets with units \cite{heijltjes:houston:14}: 
permutations of $\Wrule$ may re-assign which $\ldia$ is introduced by a specific $\Kdrule$ as in the following example.
	{$$
	\vlderivation{\vlin{}{\Wrule}{{\gclr \ldia} B , {\iclr\ldia} C \vdash {\gclr\ldia} (a \imp a) }
		{\vlin{}{\Kdrule}{{\gclr \ldia} B \vdash {\gclr \ldia} (a\imp a)}{
				\vlin{}{\Wrule}{B\vdash a\imp a}{\vlin{}{\rimprule}{\vdash a \imp a}{\vlin{}{\AXrule}{a\vdash a}{\vlhy{}}}}
			}
	}}
	\quad\not\simeq\quad
	\vlderivation{\vlin{}{\Wrule}{{\gclr\ldia} B , {\iclr\ldia} C \vdash {\iclr\ldia} (a \imp a) }
		{\vlin{}{\Kdrule}{{\iclr\ldia} C \vdash {\iclr\ldia} (a\imp a)}{
				\vlin{}{\Wrule}{C\vdash a\imp a}{\vlin{}{\rimprule}{\vdash a \imp a}{\vlin{}{\AXrule}{a\vdash a}{\vlhy{}}}}
			}
	}}
	$$}

\textbf{Winning strategies for linear logic}

We foresee no difficulties in defining \WISs for \emph{elementary} and \emph{light} linear logic adapting the techniques used for defining \CPs for multiplicative and exponential linear logic in \cite{acc:EHPN}.

We can envisage an encoding of $\oc A$ of the form $v \gmod \arof A $ for vertex $v$ such that $\lab v=\oc$.
This would avoid the need of defining the arena of $\oc A$ as the tensor of infinitely many copies of $A$, that is $\oc A=A \ltens A \ltens \cdots $, preventing the need of a quotient on \WISs required to capture the natural isomorphism between the copies of $A$.

In particular, to recover the results of Murawski-Ong for light linear logic \cite{mur:ong:evolving}, it suffices to consider the modalities $\oc$ and $\S $ as instances of $\lbox$, to define a frame condition simplifying the one of $\CK$-frames (since there are no $\ldia$), and restrain {skew-fibration} allowing deep weakening and deep contraction only to $\oc$-formulas using techniques similar to the ones in \cite{acc:EHPN}.


\bibliographystyle{plain}
\bibliography{refICPK}

\end{document}